\documentclass[11pt,draftcls,onecolumn]{IEEEtran} 
\usepackage{amsmath, amsthm, amssymb}
\usepackage{bm}
\usepackage{url}
\usepackage{cite}
\usepackage{graphics,graphicx}
\usepackage{verbatim}
\usepackage{enumerate}
\usepackage{subfigure}

\newtheorem{thm}{Theorem}
\newtheorem{cor}[thm]{Corollary}
\newtheorem{lem}[thm]{Lemma}
\newtheorem{prop}[thm]{Proposition}
\theoremstyle{remark}
\newtheorem{rem}{Remark}
\theoremstyle{definition}
\newtheorem{defn}{Definition}
\theoremstyle{definition}
\newtheorem{ass}{Assumption}

\newcommand{\ep}{\epsilon}
\newcommand{\E}{\mathbb{E}}
\newcommand{\R}{\mathbb{R}}

\newcommand{\hide}[1]{}
\DeclareGraphicsExtensions{.pdf, .png}

\begin{document}
\title{Model-based clock synchronization protocol}

\author{Nikolaos~M.~Freris,~\IEEEmembership{Member,~IEEE,}
        Vivek~S.~Borkar,~\IEEEmembership{Fellow,~IEEE,}
        and~P.~R.~Kumar~\IEEEmembership{Fellow,~IEEE}%
\thanks{N. Freris is with the School of Computer and Communication Sciences, \'{E}cole Polytechnique F\'{e}d\'{e}rale de Lausanne, EPFL, Station 14, IC/LCAV, BC 322, Lausanne CH-1015, Switzerland, E-mail: {\tt nikolaos.freris@epfl.ch}}%
\thanks{V. Borkar is with the Department of Electrical Engineering, Indian Institute of Technology Bombay, Mumbai 400076, India, 
Email: {\tt borkar.vs@gmail.com}}%
\thanks{P. R. Kumar is with the Department of Electrical and Computer Engineering at Texas A\&M University, Room 331E, Wisenbaker Engineering Research Center,  College Station TX 77843-3259, USA. E-mail: {\tt prk@tamu.edu}.}
\thanks{This work was submitted to IEEE/ACM Transactions on Networking--Nov. 2013.}
}
\markboth{TECHNICAL REPORT}
{Freris et al.: Model-based clock synchronization protocol}

\maketitle

\begin{abstract}
In a network of communicating motes, each one equipped with a distinct clock, a given node is taken as reference and is associated with the time evolution $t$. We introduce and analyze a stochastic model for clocks in which the relative speedup of a clock with respect to the reference node, called the \emph{skew}, is characterized by some parametrized stochastic process. We study the problem of synchronizing clocks in a network which amounts to estimating the instantaneous relative skews and relative \emph{offsets}, i.e.,  the differences in the clock readouts, by asynchronously communicating time-stamped packets between neighboring nodes.

For a given communication link, we develop an algorithm for filtering the time-stamps exchanged between the two nodes in order to obtain Maximum-Likelihood (ML) estimates of the logarithm of the relative skew of one clock with respect to the other. We highlight implementation issues of the optimal filter and further present a scheme for pairwise offset estimation based on the obtained relative skew estimates. We study the properties of the proposed algorithms and provide theoretical guarantees on their performance.

We extend the analysis to the problem of network-wide synchronization where the goal is to obtain, for each node and any given time instant, estimates of its skew/offset with respect to the reference time, as well as estimates of its relative skew/offset with respect to any other clock in the network.  We leverage Kalman-Bucy filtering to obtain ML estimates of the logarithm of the skews at any given time; this gives rise to an online asynchronous algorithm for optimal filtering of the time-stamps in the entire network, which is however \emph{centralized}. To tackle this, we propose an efficient \emph{distributed} suboptimal online estimator for network-wide synchronization of skews and offsets, and study its performance both analytically and experimentally. We summarize our findings into defining a new distributed model-based clock synchronization protocol (MBCSP), and present a comparative simulation study of its accuracy versus prior art to showcase improvements.
\end{abstract}

\section{Introduction}

In a large network of motes, different agents need to take actions in order for the network to perform a collaborative task as a whole. The typical paradigm is that of a \emph{wireless sensor network} (WSN)~\cite{hemant_nick} in which nodes possessing sensing, wireless communication, and computation capabilities are deployed over a large area; the sensors communicate with their neighbors wirelessly, and cooperate with each other in processing the data. Applications are ubiquitous including a) \emph{military} e.g. monitoring forces, battlefield surveillance, targeting, damage assessment and attack reconnaissance, b) \emph{environmental} such as habitat monitoring, fire detection, traffic control, seismography and agriculture research, c) \emph{health} for instance monitoring patients and controlling the drug use and d) \emph{smart homes} where sensors can help automate decisions about temperature and lighting control as well as for surveillance and safety purposes.

In such setup, centralized coordination is not an option for various reasons: first, because centralized computations typically require multi-hop communication across an identified spanning tree, which in turn implies that nodes in the proximity of a central processor will be especially burdened. More importantly, because nodes have limited energy budget such approach would lead to reduced network life-time. Last but not least, centralized computations are not robust to dynamic changes in network topology resulting from temporary link or node failures due to various reasons including mobility, fading or time-outs due to congestion. In the absence of centralized coordination, it becomes vital to achieve network-wide time synchronization among different agents.

Different clocks generally don't agree and this creates issues in determining the notion of time in a system-wide consistent fashion. The clocks in a sensor network are typically inconsistent because of frequency drifts in oscillators due to environmental conditions, such as temperature, pressure or battery voltage changes~\cite{nfrer_TAC}. In addition, some events on specific sensors may affect the clock evolution; for example, a Berkeley mote sensor may miss clock interrupts when busy handling message transmission or sensing tasks~\cite{avg_consensus}.

Many applications have stringent clock synchronization requirements. Closing control loops or coordinating events in a decentralized system~\cite{convergence}, tracking, surveillance, target localization~\cite{plarre}, data fusion, and scheduled operations like power-efficient duty-cycling are grounded in accurate time synchronization. Slotted communication protocols such as slotted ALOHA, TDMA scheduling or random-access MAC~\cite{robert} require accurate synchronization in order to avoid unnecessary resource wastage. As we head towards the era of event-cum-time driven systems featuring the convergence of computation and communication with control, the need for well-synchronized clocks becomes increasingly important, affecting all quality-of-service (QoS), system performance or even safety, since un-coordinated actions can result in destabilizing critical networked control systems.

A commonly asked question about clock synchronization in sensornets is ``why not use GPS?'' Global Positioning System (GPS) has been a widespread means of synchronizing to a universal time standard. However, such centralized synchronization scheme is unapplicable in sensor network applications for various reasons. First, GPS requires a clear sky view so it does not work inside buildings, underwater, beneath dense foliage, in densely populated downtown areas with tall buildings, or even during solar flares~\cite{sn_survey}. Many sensornet motes, including the popular Berkeley mote~\cite{TPSN}, do not come equipped with a GPS receiver because of complexity and energy issues, cost and size factors~\cite{sn_survey}; the cost of a GPS receiver is multiple times that of the cheap off-the self sensors envisioned for deploying very large WSNs of reasonable cost, while its size is also larger than the mote itself. In addition, connecting to the GPS service is too energy-consuming, multiple times than that of performing computation or wireless communication, and it can drain a sensor's battery very fast. Finally, in adversarial environments, the GPS signals may not be trusted.

Given that centralized synchronization is not an alternative, the only option it to design and implement distributed clock synchronization algorithms for WSNs. The synchronization has stringent requirements in terms of high energy efficiency, scalability, precision (typically in the order of microseconds), robustness to node and link failures, responsiveness to clock drifts as well as low communication and computation complexity.

\subsection{Problem setup}

There is no universal definition of clock synchronization, so one needs to specify the \emph{goal} of a synchronization algorithm.
There are three main notions of clock synchronization that may be required by a particular application~\cite{nfrer_TAC}: first, there is
\emph{ordering of events}, where the goal is to create a consistent chronology of events in the entire network, while specification of exact time instants is not required. Then there is \emph{relative synchronization}, aiming at estimating the time display differences among a set of clocks in the network; such information can be used to translate time-stamps from one clock to the units of any other clock. Finally there is \emph{absolute synchronization}, seeking to set clock displays to agreement.

Another classification of synchronization algorithms is by \emph{scope}. \emph{Local} synchronization refers to the problem of synchronizing a node with only a subset of other nodes, typically those in its physical proximity; this is sufficient for several applications like monitoring, tracking and surveillance. For applications where the target is to achieve network-wide cooperation e.g., duty-cycling, slotted protocols or network control, it is necessary to attain \emph{global} synchronization which involves \emph{all} nodes in the network.
In this paper, we address the most general problem of \emph{absolute global} clock synchronization where each node calculates estimates of its skew/offset with respect to the reference clock, and can further compute estimates of its relative skew/ estimate w.r.t. any other node in the network, let alone any of its neighbors.


To achieve network clock synchronization, nodes are allowed to exchange time-stamped packets with their neighbors, i.e., with all nodes within their communication range. Such packet exchanges undergo a random communication delay that is unknown and time-varying. Here by ``delay'' we do not just mean the electromagnetic propagation delay, but rather the total elapsed time between the time that a sending node time-stamps a packet until the time that a receiver decodes it~\cite{nfrer_TAC}. At the transmitter side, this typically includes the time for the operating system to process the packet after time-stamping it, followed by the time for the packet to make its way through the remainder of the communication protocol stack. Then we have electromagnetic propagation delay which is totally determinable by the distance between sender and receiver; this is in the order of a few nanoseconds for the typical case when nodes are a few meters apart. At the receiver end, there is the time it takes the packet to travel through the protocol stack plus the time taken to consult the local clock on the receipt time. Because of the randomness at the transmitter/sender side due to e.g., operating system interrupts and congestion at the MAC layer, delays are random and asymmetric in the two directions of a communication link~\cite{nfrer_TAC}. Furthermore, the trasnmitter/sender delays dominate the propagation delay which is also the only deterministic component of the total delay incurred to a packet.

The ultimate goal is to develop a distributed asynchronous algorithm for the optimal filtering of time-stamps with provable performance.
In the sequel, we show that \emph{delay estimation} can be viewed a natural bi-product of clock synchronization. Estimating communication delays has many applications in its own right e.g., in predicting receipt times for network control and scheduling operations.

\subsection{Main Results}

The specific contribution that we outline in this paper is the introduction of a mathematical model for clocks and its use for developing a new clock synchronization protocol. 
This paper expands and extends preliminary results that appeared in~\cite{nfrer_algocdc}.

In a network of $n+1$ nodes, each equipped with its own clock, we develop a mathematical model for the time displays of different clocks and analyze its properties. We denote the display of node $i$'s clock at time $t$ by $\tau_i(t)$; with no loss in generality, we assume that the time evolution $t$ corresponds to node $0$'s clock, i.e., we define node 0 as the reference. We call the relative instantaneous speedup of clock $i$ at reference time $t$ with respect to the reference clock as its ``\emph{skew},'' and denote it by $a_i(t)$. The ``\emph{offset}" of clock $i$ at time $t$ is defined as the instantaneous difference of its display with respect to the reference time $t$, i.e., $\tau_i(t) - t$. For a given clock $i$, we model its skew $a_i(t)$ as a stochastic process, given by an exponential transformation of an Ornstein-Uhlenbeck process. Its time display, $\tau_i(t)$, is given by the integral of the skew $a_i(s)$ in the interval $[0,t]$. For each node, the instantaneous skew has expected value 1 at all times, while its variance is bounded by a constant; this is in alliance with  experimental observations showcasing that the skew of real clocks tends to fluctuate from the nominal value in a rather controlled manner~\cite{veitch}. Furthermore, the time display is \emph{unbiased}, i.e., $\E[\tau_i(t)] = t$ but, nonetheless, has unbounded variance which grows with time. This unbounded variance corresponds to the physical accumulation of skew errors with time
~\cite{veitch} and makes the synchronization problem challenging.
We show that if a different clock is taken as reference, the time displays of all other clocks can be expressed with respect to it using the same model, with different parameters and a change of time scales. We also calculate the Allan variance \cite{allan} of the model, which is a widely used metric for skew stability.

Under an assumption of sufficiently frequent packet exchanges and slowly-varying delays, we show how to obtain noisy relative skew measurements from one-way exchange of time-stamps in a communication link. We proceed to study the problem of pairwise synchronization, i.e., minimum-variance unbiased estimation of the logarithm of the relative skew, as an instance of linear filtering~\cite{jaswinski}. We further establish properties of and bounds on the error variance.
We leverage the analysis for pairwise clock synchronization to network clock synchronization. In such a case, the differential equations that are derived for filtering are not readily implementable since they require integration with respect to the unknown reference time. To tackle this, we propose an approximation which still leads to a bounded variance filter.

We propose and analyze three schemes for the network-wide estimation of the skews. First, we consider an off-line algorithm for the filtering of pairwise estimates. We show that using the distributed scheme of~\cite{arvind,barooah_conf} for the smoothing of pairwise estimates is optimal for the network-wide estimation for a particular selection of error criterion. 
Second, we show that the optimal linear filtering equations for the network case yield an asynchronous scheme which is stable and of bounded variance, but in general requires centralized computations. Last and more importantly, we derive a suboptimal asynchronous scheme which is readily implementable in a decentralized fashion.

We  address protocol implications and present a scheme to estimate relative offset estimates based on estimates of the relative skews.
We summarize our findings into defining \emph{model-based clock synchronization protocol} (MBCSP), and further present a comparative simulation study with both synthetic and real data in order to showcase improvement over previous art.

\section{Related Literature}

The simplest yet most commonly adopted model for clocks is the so called \emph{affine} model. In this model, the skew is assumed \emph{constant}, at least for the duration of a synchronization period; in practice, this implies that synchronization has to be carried out at faster time-scale than the rate of clock drifts. The display of the $i-$th clock satisfies $\tau_i(t) := a_i(t-t_0) + b_i,$ where $t$ denotes the reference time evolution, $a_i > 0$ is the skew of clock $i$, and $b_i \in \R$ is the offset of the clock at reference time $t_0$. A complete analysis of the feasibility of the problem for the case of \emph{affine} clocks was carried out in~\cite{nfrer_TAC}. It was shown that clock synchronization is \emph{impossible} even for the best case scenario of a complete graph and constant (but unknown) link delays, and the uncertainty set was fully characterized. 
The majority of existing synchronization protocols considers constant skew for fixed time periods, and accounts for skew drifts by adaptively estimating $a_i$ for each given time interval.

For a comprehensive survey on clock synchronization with emphasis on sensor network applications we refer the reader to~\cite{sn_survey2,sn_survey,sync_overview}. In the sequel, we overview related work featuring popular synchronization protocols, as well as theoretical results on the topic.

In large networks where synchronization requirements are not too stringent, e.g. the Internet, the Network Time Protocol
(NTP~\cite{NTP}) has been successfully used for over two decades. NTP is a \emph{hierarchical} protocol used to synchronize with external sources (typically via GPS) organized in a hierarchy of levels, called \emph{strata}. It attains accuracy in the order of a few
milliseconds~\cite{NTP}. However, real-time applications in wireless sensor networks typically require precision in the order of a few microseconds, so NTP is no longer applicable in this setup.

For accurate synchronization in sensor networks a variety of algorithms have been suggested.
Reference Broadcast Synchronization (RBS~\cite{RBS}) is a \emph{receiver-receiver} synchronization algorithm which exploits
the broadcast nature of the wireless medium. In RBS, nodes broadcast packets without any time-stamping on the
transmitter side. The  nodes that receive the same transmitted packet record the reception times and exchange this information with
their neighbors so as to estimate their clock difference; an implicit assumption is that the one-way delays are the same for
neighboring nodes, which completely eliminates the transmitter-side non-determinism and accuracy depends
mainly on the difference of receiver processing delays. This scheme was tested in sensor networks comprising of
Berkeley motes using off-the-shelf 802.11 Ethernet, and achieved precision within 11 $\mu s$~\cite{RBS}.
The Timing-Sync Protocol for Sensor Networks (TPSN~\cite{TPSN}) is a popular \emph{sender-receiver} synchronization protocol which
uses time-stamping at the MAC layer to eliminate the transmission delay which is typically the most variable term in
wireless sensor networks. It was observed and verified by simulations~\cite{TPSN} that TPSN achieves approximately
twice the accuracy of RBS for pairwise synchronization. TPSN uses message-passing across a spanning tree to achieve network-wide synchronization. In~\cite{FTSP}, authors identify the sources that contribute to packet delivery delay and propose the Flooding Time Synchronization Protocol (FTSP~\cite{FTSP}). FTSP uses hardware solutions and an efficient time-stamping mechanism to effectively eliminate all packet delay factors but the propagation delay. Linear regression is used to compensate for clock drifts, and the achieved precision was measured 10 $\mu s$ on average, for a network with several hundreds of nodes. Moreover, FTSP was found to be highly robust to link failures, and dynamic readjustments of the reference node~\cite{FTSP}.

A widely studied problem is that of obtaining \emph{consistent} estimates for clock's offsets, given estimates of relative measurements between any two communicating nodes~\cite{karp,solis,arvind,smoothing,barooah_conf,barooah,barooah_scalinglaws,barooah_asymmetric,barooah_MJLS}; a comprehensive survey can be found in~\cite{barooah_survey}. The term ``consistent'' means that estimates should satisfy Kirchoff's voltage law, i.e., should sum to zero along the edges of a directed cycle. This constraint leads to \emph{smoothing} pairwise measurements and is also used for obtaining nodal skew estimates by considering relative measurements of the logarithm of clock's skew. Karp et al.~\cite{karp} derived the \emph{best linear unbiased estimate} (BLUE) of the pairwise offset differences between any two nodes $(i,j)$, where ``best'' refers to minimum-variance. They further established that the variance of the BLUE is equal to the effective \emph{resistance} between those two nodes in an electric network where the resistance of each link is equal to the error variance of the corresponding relative measurement. The connection of BLUEs with resistance networks has also been thoroughly studied in~\cite{arvind,smoothing,barooah_scalinglaws}; the authors therein studied the asymptotic variance for large networks and established sharp scaling laws for various classes of graphs. It follows  from the electrical analogy that there is a great merit in leveraging all relative measurements not just those along a spanning tree as in TPSN~\cite{TPSN}.  For example, for connected random planar networks the asymptotic variance of the BLUE between any two nodes is $O(1)$~\cite{smoothing}, while for a tree it is $O(\sqrt{n})$; this provides theoretical support for the feasibility of clock synchronization in large networks.

A large emphasis has been given on developing distributed asynchronous schemes for iteratively calculating the BLUEs by \emph{smoothing} relative measurements. One such scheme was developed in~\cite{karp} under the restrictive assumption of Gaussian measurement noise. An asynchronous version of the Jacobi algorithm, which essentially amounts to computing the pseudo-inverse of the weighted reduced Laplacian of the network graph in a decentralized fashion, was developed simultaneously and independently in~\cite{barooah_conf,solis}; The achieved average accuracy was measured to be 2$\mu s$ for pairwise synchronization, and $20 \mu s$ for a network of $40$ nodes, in which FTSP achieved an average accuracy of $30 \mu s$~\cite{solis}. In~\cite{barooah}, it was proposed to use information from the two-hop neighborhood of each node for deriving a distributed algorithm called \emph{Overlapping Subgraph Estimator}; the convergence to BLUE was established and it was argued that this scheme has a faster convergence speed, therefore leading to energy savings in a sensor network. Barooah et al.~\cite{barooah} also proposed an efficient method for initializing the nodal estimates, called \emph{flagged initialization}, which led to reducing the number of iterations (for a given convergence criterion) in all tested scenarios. Giridhar and Kumar~\cite{arvind,smoothing} obtained bounds on the convergence rate of the \emph{spatial smoothing} algorithm, i.e., the Jacobi algorithm applied to least-squares estimation. They made connections of the convergence rate with two graph indices, namely the degree of the root node and the \emph{edge-connectivity}, by using Cheeger's inequality. For example, for either lattices or random planar graphs, it was shown that the number of iterations required to attain a given accuracy is $O(n^2)$~\cite{smoothing}. An implicit assumption of the aforementioned approaches is that communications links are symmetric. When this assumption is violated, it was shown in~\cite{barooah_asymmetric} that the distributed Jacobi algorithm converges to a suboptimal unbiased estimate. The case where network topology is varying, e.g. due to mobility, was studied in~\cite{barooah_MJLS} where convergence was established using results from  Markov Jump Linear Systems. Distributed synchronization algorithms are directly related to consensus; algorithms based on average consensus have been proposed in~\cite{avg_consensus,dpll}, while~\cite{MTS} proposed a protocol based on maximum-value consensus. Last but not least, a class of randomized distributed smoothing algorithms with provable exponential convergence in the mean square appeared in~\cite{RKO}.



%

Further related literature treats the case of \emph{atomic} clocks.  A model for atomic clocks was proposed in~\cite{model_sde}, \cite{allan_sde}; The authors suggested linear stochastic differential equation systems where the state variables correspond to the phase and frequency deviations of the clock oscillator. As a metric for clock stability, an index of the quadratic variation of the clock frequency called \emph{Allan variance} was introduced in~\cite{allan} and was calculated for the particular clock model in~\cite{allan_sde}.

\section{Stochastic Clock Model} \label{clock_model}

In a network of $n+1$ clocks, we consider a reference time, denoted by the non-negative continuous variable $t\in \R_+$. Without any loss in generality, we assume that the reference corresponds to a given clock, say clock $0$. We use $\tau_i(t)$ to denote the display of the $i-$th clock at time
$t$.
To model such time displays $\tau_i(t)$, we define independent Ornstein-Uhlenbeck processes~\cite{cont_mart} $\{X_i(\cdot)\}_{i=1}^n$:
\begin{eqnarray}\label{state_eq}
dX_i(t) = -\alpha_iX_i(t)dt + \ep_idW_i(t) & ; & X_i(0) = 0, \label{OU}
\end{eqnarray}
where $\{W_i(\cdot)\}_{i=1}^n$ are independent scalar Brownian motions  and where $\alpha_i, \ep_i > 0$ are given constants.
Let $a_i(t)$ denote the \emph{skew} of the $i-$th node at time $t$, i.e., its relative speed with respect to the
reference clock. We model the skew as
\begin{eqnarray}\label{skew_eq}
a_i(t) := c_i(t)e^{X_i(t)}, & c_i(t):= e^{- \frac{1}{4}\frac{\ep_i^2}{\alpha_i}(1 - e^{-2\alpha_it})} 
\end{eqnarray}
Then the clock display $\tau_i(t)$ is given by
\begin{equation}\label{disp_eq}
\tau_i(t) := \int_0^t\!\!a_i(t')dt'; \ \tau_i(0) = 0.
\end{equation}

\begin{rem}
The model also applies for the reference clock $0$, by setting $\ep_0 = 0$.
We have assumed that all clocks start \emph{synchronized} in the sense that $\tau_i(0) = 0, a_i(0)=1$, and that synchronization is lost as an effect of time-varying skews; the extension to the general case is straightforward.
Note that $a_i(t) > 0 \ \forall t\ge0, \ i=0,1,\hdots,n$, which implies that for
all clocks the time displays $\tau_i(t)$ are strictly increasing functions of $t$, i.e., the time at each clock
evolves \emph{forward}.
\end{rem}

\begin{rem}[Stochastic differential equation for the skew] \label{skew_sde}
It follows as an application of It\^{o}'s formula \cite{jaswinski}
and the definition (\ref{skew_eq})
that 
the skew $a_i(t)$ satisfies the following Stochastic Differential Equation (SDE):
\begin{equation}\label{sde_a}
da_i(t) = (-\alpha_i \log a_i(t) + \frac{1}{4}\ep_i^2(3-e^{-2\alpha_i t}))a_i(t)dt + \ep_i a_i(t) dW_i(t).
\end{equation}
\end{rem}

\begin{rem}[Numerical simulation of OU process]\label{OU_sim}
The Ornstein-Uhlenbeck process $X_i(t)$ can be numerically approximated by a discrete process $\bar{X}_i(k)$ defined recursively by
\begin{equation}\label{numerical_OU}
\bar{X}_i((k+1)) = (1 - \alpha_i \Delta t) \bar{X}_i(k) + \ep_i \sqrt{\Delta t} w_i(k),
\end{equation}
where $\Delta t$ is the stepsize, $\bar{X}_(k)$ denotes the numerical approximation of $X_i(k\Delta t)$ and $\{w_i(k)\}$ represents a sequence of independent standard normal random variables. For the linear SDE under consideration (OU), this method corresponds to both Euler-Maruyama and Milstein's methods~\cite{numerical}. In particular the numerical approximation (\ref{numerical_OU}) has \emph{strong order of convergence} $\gamma = 1$, in the sense that for a fixed interval $[0,T]$ there exists $C(T)>0$ such that
\begin{equation}
\E|\bar{X}_i(n) - X_i(n\Delta t)| \le C(T)\Delta t^\gamma,
\end{equation}
for all $n=0,1,\hdots, \lfloor\frac{T}{\Delta t} \rfloor$.

However, for the OU case, an exact calculation is possible at discrete values as follows. From the analysis in appendix \ref{prop_app} it follows that for any $t\ge0$ we have
\begin{equation}
X_i(t) = \ep_i\int_0^t \! \! e^{-\alpha_i(t - s)}dW_i(s),
\end{equation}
whence for $t_2\ge t_1 >0$
\begin{equation}
X_i(t_2) = e^{-\alpha_i(t_2-t_1)} X_i(t_1) + \int_{t_1}^{t_2} \! \! e^{-\alpha_i(t_2 - s)}dW_i(s),
\end{equation}
in particular we can consider again a uniform sampling with step $\Delta t$ and get
\begin{equation}\label{numerical_OU_exact}
\bar{X}_i(k+1) =  e^{-\alpha_i\Delta t} \bar{X}_i(k) + \sqrt{\frac{1}{2\alpha_i}[1 - e^{-2\alpha_i\Delta t}]}v_i(k),
\end{equation}
where $v_i(k)$ is standard (mean zero and unit variance) Gaussian White Noise (GWN) sequence.
\end{rem}

\subsection{Model Properties} \label{prop}

In this section, we summarize the main properties of the proposed mathematical model. In particular, we show that, for each clock, the
instantaneous skew has mean value $1$ at all times, and bounded variance. However, the variance of the time display grows unbounded.
We further show how to express the time of one clock with respect to the time of another, which comes handy in the case that the reference
may be dynamically re-selected to account for changes in network topology such as broken links/nodes.

\begin{lem}[Properties of the stochastic model]\label{prop_lemma}
For each clock $i=1,2,\hdots,n$:
\begin{enumerate}
\item The skew satisfies $\E[a_i(t)] = 1 , \forall t\in \R_+$, and $\sup_{t\in\R_+} Var(a_i(t)) < +\infty$.\label{skew_prop}
\item The time display satisfies $\E[\tau_i(t)] = t , \lim_{t \to +\infty} Var(\tau_i(t)) = +\infty$; in particular, $Var(\tau_i(t)) = \Omega(t), \ Var(\tau_i(t)) = O(t^2)$. \label{time_prop}
\end{enumerate}
\end{lem}
\begin{proof}
The proof can be found in Appendix \ref{prop_app}.
\end{proof}

\begin{rem}[Metric for clock's quality]
One may use either $\frac{\alpha_i}{\ep_i^2}$ or $\frac{\alpha_i}{\ep_i}$ as a metric for the ``\emph{quality}" of the $i$-th
clock. This is justified by noting that the asymptotic skew variance as well as the upper bound on the variance of the time display of clock $i$, is increasing in $e^\frac{\ep_i^2}{\alpha_i}$, while the lower bound on $Var(\tau_i(t))$ is increasing in $\frac{\alpha_i}{\ep_i}$, cf. (\ref{varskew_bound}), (\ref{vartime_bound}) in appendix \ref{prop_app}. In conclusion, the expected deviation of the clock from nominal values is decreasing in each metric.
\end{rem}
The following result characterizes the sample path behavior of the skew process $a_i(t)$.
\begin{thm}[Sample path properties of skew]
The skew $a_i(t)$ almost surely satisfies
\begin{eqnarray}
\overline{\lim}_{t\to +\infty} a_i(t) &=& \sup_{t\ge0} a_i(t) = +\infty\\
\underline{\lim}_{t\to +\infty} a_i(t) &=& \inf_{t\ge0} a_i(t) = 0.
\end{eqnarray}
\end{thm}
\begin{proof}
It follows from (\ref{sol}) that $Y_i(t) := e^{\alpha_i t}X_i(t) = \ep_i\int_0^t \! \! e^{\alpha_i
s}dW_i(s)$ is a continuous local martingale \cite{cont_mart}, with quadratic variation $[Y_i,Y_i]_t = e^{2\alpha_it} -
1$, and $Y_i(0) = 0$. Hence the Dambis, Dubin-Schwartz theorem~\cite{cont_mart}, applies and yields that
$Y_i(t) = \tilde{W}_i([Y_i,Y_i]_t)$, where $\{\tilde{W}_i(t)\}$ is a standard Brownian motion; in particular, $X_i(t) = \ep_i e^{-\alpha_it}\tilde{W}_i(e^{2\alpha_it} - 1)$.
It follows from the law of iterated logarithms~\cite{stroock} which states that for Brownian motion $\{W(t)\}$ it holds that
$\overline{\lim}_{t\to +\infty} \frac{|W(t)|}{\sqrt{2t\log\log t}} = 1, a.s.$
that $\overline{\lim}_{t\to +\infty} X_i(t) = +\infty, \underline{\lim}_{t\to +\infty} X_i(t) = -\infty, a.s.$
\end{proof}

%

A particular realization of the clock display according to our model, as well as a plot of its variance are illustrated in Section~\ref{sec:simulations}.
\subsection{Time translation among different clocks}

We show how to translate time among different (non-reference) clocks $i$ and $j$,
i.e., how to express the time of one non-reference clock in the units of another clock.

\begin{prop}[Time translation]\label{translation_prop}
Assume that $\alpha_i=\alpha_j=\alpha$, and consider the time of the $j$-th clock  expressed with respect to the time of the
$i$-th clock; formally let $\tilde{\tau}_{ji} := \tau_j \circ \tau_i^{-1} $, i.e., $\tilde{\tau}_{ji}(\tau) :=
\tau_j(\tau_i^{-1}(\tau))$.  Then
\begin{enumerate}
\item $\tilde{\tau}_{ji}(\tau)$ satisfies the differential equation
    \begin{equation}\label{time_tran}
    d \tilde{\tau}_{ji}(\tau) = \frac{a_j(t)}{a_i(t)}d\tau,
    \end{equation}
    where $t = \tau_i^{-1}(\tau)$.

\item For $i,j \ne 0$, let $X_{ij}(t) := X_j(t) - X_i(t)$. Then $X_{ij}(t)$ satisfies the SDE
    \begin{equation}\label{OU2}
    dX_{ij}(t) = -\alpha X_{ij}(t)dt + \ep_{ij}dW_{ij}(t) \ ; \ X_{ij}(0) = 0,
    \end{equation}
    and the ``\emph{relative skew}'' is given by
    \begin{equation}\label{relative_skew}
    a_{ij}(t) := \frac{a_j(t)}{a_i(t)} = c_{ij}(t)e^{X_{ij}(t)},
    \end{equation}
    where $ c_{ij}(t):= \frac{c_j(t)}{c_i(t)} = c_{ij} e^{\frac{1}{4}(\frac{\ep_j^2}{\alpha} -
\frac{\ep_i^2}{\alpha})e^{-2\alpha t}}$, $c_{ij} := e^{-\frac{1}{4}(\frac{\ep_j^2}{\alpha} -
\frac{\ep_i^2}{\alpha})}$, $\ep_{ij} := (\ep_i^2 + \ep_j^2)^{\frac{1}{2}}$ \footnote{Note that in the case that $j=0$ we
should set $\ep_{i0} = -\ep_i$, but this would not alter the distribution of the solution of the SDE.}, and $W_{ij}(t)$
is a standard scalar Brownian motion dependent on $W_i(t), W_j(t)$.
\item Let $\tilde{X}_{ij}(\tau) := X_{ij}(\tau_i^{-1}(\tau))$. Then $\tilde{X}_{ij}(0) = 0$ and
    \begin{eqnarray}\label{state_eq2}
    d\tilde{X}_{ij}(\tau) = -\alpha \frac{1}{a_i(t)}\tilde{X}_{ij}(\tau)d\tau +
\ep_{ij}\frac{1}{\sqrt{a_i(t)}}dW_{ij}(\tau),
    \end{eqnarray}
    where $t:= \tau_i^{-1}(\tau_i)$.
\end{enumerate}
\end{prop}

\begin{proof}
The function $\tau_i(t)$ is strictly increasing and continuous on $\R_+$, and tends to infinity. Therefore, it is bijective so the inverse $\tau_i^{-1}(\tau)$ exists and
is also strictly increasing. By the chain rule, we have that $\frac{d}{d\tau}\tau_j(\tau_i^{-1}(\tau)) =
\dot{\tau_j}(t)\frac{d}{d\tau}\tau_i^{-1}(\tau) = a_j(t)\frac{1}{\dot{\tau_i}(t)} = \frac{a_j(t)}{a_i(t)}$, where  
$t := \tau_i^{-1}(\tau)$ and we used the notation  $\dot{f}(t) := \frac{d}{dt}f(t)$.
To prove the second part, note that the SDE (\ref{state_eq}) is simply a shorthand for the stochastic integral equation
\begin{equation}
X_i(t) = -\alpha_i \int_0^t X_i(t')\,dt' + \ep_iW_i(t).
\end{equation}
Then, since $\alpha_i=\alpha_j=\alpha$, it follows that $\frac{a_j(t)}{a_i(t)} = c_{ij}(t)e^{X_{ij}(t)}$ and the result
follows from
\begin{equation}
X_{ij}(t) = -\alpha \int_0^t X_{ij}(t')\,dt' + \ep_jW_j(t) - \ep_iW_i(t),
\end{equation}
and Levy's characterization of Brownian motion \cite{cont_mart}.
The last part follows by a time-scale change in the equation
\begin{equation}
X_{ij}(t) = -\alpha \int_0^t X_{ij}(t)\,dt + \ep_{ij}W_{ij}(t).
\end{equation}
Setting $\tilde{X}_{ij}(\tau) := X_{ij}(\tau_i^{-1}(\tau))$ we get
\begin{equation}
\tilde{X}_{ij}(\tau) = -\alpha \int_0^{\tau_i^{-1}(\tau)} X_{ij}(\tau')\,d\tau' + \ep_{ij}W_{ij}(\tau_i^{-1}(\tau)).
\end{equation}
Using the change of variable $t' = \tau_i^{-1}(\tau')$  we get that the drift term $-\alpha \int_0^{\tau_i^{-1}(\tau)}
X_{ij}(\tau')\,d\tau'$ becomes $-\alpha \int_0^{t} \frac{1}{a_i(t')}\tilde{X}_{ij}(t')\,dt'$.
Since $a_i(t)$ is characterized by a bijective transformation of $X_i(t)$ it is progressively measurable with respect to
the filtration $\sigma(W_{ij}(t)) := \sigma(W_i(t))\bigvee\sigma(W_j(t))$ \cite{stroock}. Hence it follows that $M_\tau
:= \int_0^{\tau} \frac{1}{\sqrt{a_i(t')}}dW_{ij}(\tau')$, for $t' = \tau^{-1}(\tau')$, is a continuous local martingale
\cite{cont_mart}, $M(0) = 0$, and its quadratic variation is $[M,M]_\tau = \tau_i^{-1}(\tau)$. By the celebrated Dambis,
Dubin-Schwartz theorem \cite{cont_mart} $M_\tau = W_{ij}([M,M]_\tau) =
W_{ij}(\tau_i^{-1}(\tau))$.
\end{proof}
\begin{rem}[Relative skew properties]\label{skew_properties}
From (\ref{OU2}), (\ref{relative_skew}), (\ref{cov_X}) it follows that it is not true that $\E[a_{ij}(t)] = 1$ nor that
$\E[a_{ij}(t)]\E[a_{ji}(t)] = 1$:
\begin{eqnarray}
\E[a_{ij}(t)] &=& c_{ij}(t) e^{\frac{1}{2}\E[X_{ij}^2(t)]} = e^{\frac{\ep_j^2}{2a}(1-e^{-2at})} > 1 ,\\
\E[a_{ij}(t)]\E[a_{ji}(t)] &\nearrow& e^{\frac{\ep_i^2 + \ep_j^2}{2a}} > 1.
\end{eqnarray}
\end{rem}

\begin{rem}[Convention]
Since the ``quality'' of a clock $i$ is characterized by the quantity $\frac{\alpha_i}{\ep_i^2}$ or $\frac{\alpha_i}{\ep_i}$, while in order to
define the translation equations for the nodes $i,j$ we needed $\alpha_i = \alpha_j$, we henceforth adopt for simplicity the convention that
$\alpha_i \equiv \alpha$ for all $i=1,\cdots,n$. The ``quality'' is then modeled solely by $\ep_i$, so we still allow for diverse clocks. Recall that $\alpha$ models the transient behavior, cf. (\ref{cov_X}).
\end{rem}

\begin{cor}[$L^p-$boundedness]\label{bounded}
For any $p\ge1$, the processes $X_i(t),X_{ij}(t)$ are $L^p-$ bounded Gaussian processes; the processes $a_i(t), a_{ij}(t), \frac{1}{a_i(t)}, \frac{1}{a_{ij}(t)}$ are also $L^p(P)-$ bounded.
\end{cor}

\begin{proof}
The case where $i=j=0$ is trivial, since then $X_0(t) \equiv 0, a_0(t) \equiv 1$.
The Ornstein-Uhlenbeck SDE is a linear SDE, hence the processes $X_i(t),X_{ij}(t)$ are Gaussian, mean $0$, and by
(\ref{cov_X}), they have uniformly bounded variance.
The rest follows from (\ref{skew_eq}), (\ref{relative_skew}).
\end{proof}

\subsection{Allan variance}

Allan variance is as a performance index for clock stability. Consider the average skew of the $i$-th clock in the
interval $[t-T,t]$:
\begin{equation}\label{mean_skew}
\bar{a}_i(t) := \frac{1}{T}\int_{t-T}^{t}a_i(t')\,dt' \ = \frac{\tau_i(t) - \tau_i(t-T)}{T}.
\end{equation}
The Allan variance $\sigma_a^2(T)$ is defined by~\cite{allan}
\begin{equation}
\sigma_{a_i}^2(T) := \frac{1}{2} \lim_{T' \rightarrow \infty}\frac{1}{T'} \int_{0}^{T'} (\bar{a}_i(t+T) -
\bar{a}_i(t))^2\,dT',
\end{equation}
provided that the limit exists.

\begin{rem}[Numerical approximation of Allan variance]\label{appr_allan}
In practice, the Allan variance is approximated based on $N+1$ periodic measurements of clock $i$'s display with period
$T$ by:
\begin{equation}
\sigma_{a_i}^2(T) \simeq \frac{1}{2N}\sum_{k=1}^{N}(\bar{a}_i((k+1)T) - \bar{a}_i(kT))^2,
\end{equation}
where $N$ is sufficiently large.
%
In the case of the proposed stochastic model the Allan variance can be defined as
\begin{equation}\label{stochastic_allan}
\sigma_{a_i}^2(T) := \frac{1}{2} \lim_{t \rightarrow \infty} \E[(\bar{a}_i(t+T) - \bar{a}_i(t))^2],
\end{equation}
where we have used the ergodicity of the OU process~\cite{cont_mart}. Using (\ref{stochastic_allan}), (\ref{mean_skew}) and (\ref{skew_cor}) we get
\begin{equation}\label{allan_model}
\sigma_{a_i}^2(T) = \frac{1}{T^2} (\int_0^T \int_0^T e^{\frac{\ep_i^2}{2\alpha_i} e^{-\alpha_i|t-s|}}\,ds\,dt - \int_0^T
\int_{-T}^0 e^{\frac{\ep_i^2}{2\alpha_i} e^{-\alpha_i|t-s|}}\,ds\,dt).
\end{equation}
\end{rem}

\begin{rem}\label{asd}
Note that the Allan variance is not equal to the sample variance of the skew process. It is also different than
$\lim_{t\rightarrow \infty} \frac{1}{2}\E[(a_i(t+T) - a_i(t))^2]$, which for our model is equal to
$e^{\frac{\ep_i^2}{2\alpha_i}} -
e^{\frac{\ep_i^2}{2\alpha_i} e^{-\alpha_iT}}$, an increasing function of $T$. We will refer to this quantity as
``asymptotic skew difference variance''.
\end{rem}

\begin{rem}[System identification]\label{syst_ident}
Based on a numerical approximation of Allan variance through periodic measurements (as in Remark \ref{appr_allan})
parameters $\alpha_i,\ep_i$ can be estimated using standard curve-fitting optimization techniques and the analytical
formula (\ref{allan_model}).
\end{rem}

\section{Measurement model} \label{meas_sec}

We model the communication network with a graph $G = (V,E)$. Node $i$ can send
packets to node $j$ if and only if $(i,j)\in E$. We allow such pair of nodes to exchange
time-stamped packets. When a node $i$ sends its $k-$th packet, it includes a time-stamp of the sent
time according to its own clock, $s^{(k)}_i$. We allow for broadcast, i.e., for a sending packet to
be received by two neighboring nodes. When node $j$ receives this time-stamped packet it records
the time according to its clock of when it received the packet, $r^{(k)}_{i,j}$. We
assume that a packet sent over a link $(i,j)$ undergoes a delay of $d_{ij}^{(k)}$ time units (as
measured by the reference clock); recall that this includes not only the electromagnetic propagation delay, but
rather the total elapsed time between time-stamping at the transmitter side and
time-stamping at the receiver end~\cite{nfrer_TAC}. Delays are typically random and asymmetric ~\cite{nfrer_TAC}; in general,
$d_{ij} \ne d_{ji}$ and $d_{ij} \ne d_{il}$ for $(i,j),(j,i),(i,l)\in E$. Note also that the time instants when a time-stamped message is sent are typically determined based on the sender clocks, and are therefore stochastic.

To obtain a noisy measurement of the relative skew $a_{ij}$, node $i$ needs to send two consecutive time-stamped packets to
node $j$; this communication scheme is depicted in Figure \ref{ping}. The measurement model we adopt is summarized in the next theorem;
the analysis can be found in appendix~\ref{model_just_app}.

\begin{figure}
    \centering
    \includegraphics[totalheight=0.25\textheight]{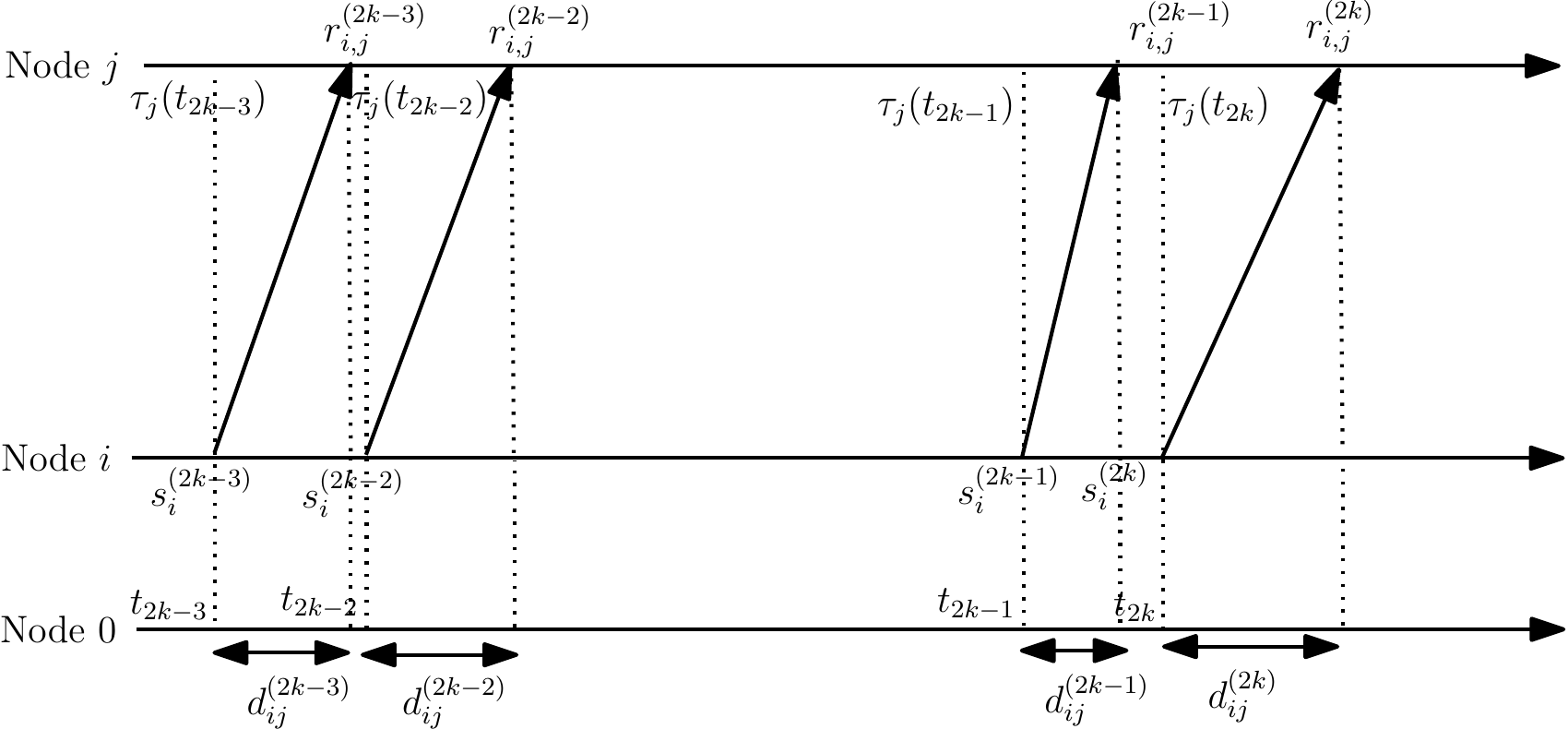}
    \caption{Communication via exchange of time-stamped packets}\label{ping}
\end{figure}

\begin{thm}[Measurement at a link $(i,j)$]\label{meas_thm}
~\\
In a link $(i,j)$, node $j$ can make a noisy measurement of the logarithm of the relative skew at time $t_k$,
$a_{ij}(t_k)$, by sending two consecutive packets at times $t_k,t_{k+1}$ and using all four send and receive
time-stamps. The measurement is of the form
\begin{equation}
y_{ij}(t_k) := X_{ij}(t_k) + v_{ij}(t_k),
\end{equation}
where \footnote{For a link $(i,j)$ such that $i,j \ne 0$, we could also ignore $e^{-2\alpha t_k}$, because neither $i$
nor $j$ has access to $t$. The error of such approximation is negligible for large values of $t$, and is dominated by
the noise term.}
\begin{equation}\label{link_meas}
y_{ij}(t_k) := \log|\frac{r_{i,j}^{(k+1)} - r_{i,j}^{(k)}}{s_i^{(k+1)} - s_i^{(k)}}| -
\frac{1}{4}(\frac{\ep_j^2}{\alpha} - \frac{\ep_i^2}{\alpha})(1-e^{-2\alpha t_k}),
\end{equation}
and $v_{ij}(t_k) \sim \mathcal{N}(0,\sigma^2_{ij}(t_k))$, is white Gaussian noise.
\end{thm}

\section{Pairwise synchronization of two clocks} \label{pairwise_sync}

We show how to filter the time-stamps in a directed link $(i,j)$ in order to derive the ML estimate of the relative skew
$a_{ij}(t)$. In this section, we use $t_k$ to denote the time instant (measured in the units of the reference time) at
which the $k$-th measurement is made
; this corresponds to $t_{2k}$ for the scheme shown in  Figure \ref{ping}.
This problem reduces to the continuous-discrete linear filtering problem~\cite{jaswinski}, for
the system
\begin{eqnarray}
dX_{ij}(t) &=& -\alpha X_{ij}(t)dt + \ep_{ij}dW_{ij}(t) \  ; \  X_{ij}(0) = 0, \label{state}\\
y_{ij}(t_k) &=& X_{ij}(t_k) + v_{ij}(t_k) \ ; \ v_{ij}(t_k) \sim \mathcal{N}(0,\sigma^2_{ij}(t_k)) \label{meas},
\end{eqnarray}
The ML estimate can be computed by the discrete Kalman-Bucy filter\footnote{Throughout, we use the standard notation for
estimates and conditional-unconditional error variances as in~\cite{jaswinski}.}
\begin{eqnarray}
d\hat{X}_{ij}(t) &=& -\alpha \hat{X}_{ij}dt  , \ \ t\ge0\ \ ; \hat{X}_{ij}(0) = 0,\label{state_est}\\
\hat{X}_{ij}(t_k^+) &=& \hat{X}_{ij}(t_k^-) + K(t_k)(y_{ij}(t_k) - \hat{X}_{ij}(t_k^-))\label{state_meas},\\
K(t_k) &=& \frac{P_{ij}^{t_k^-}(t_k)}{P_{ij}^{t_k^-}(t_k) + \sigma^2_{ij}(t_k)}\label{kalman_gain}.
\end{eqnarray}
The conditional error variance is equal to the unconditional error variance $P_{ij}(t)
:=\E[(X_{ij}(t) - \hat{X}_{ij}(t))^2] = \E[P_{ij}^t(t)]$~\cite{jaswinski}. It satisfies
\begin{eqnarray}
\frac{dP_{ij}^t(t)}{dt} &=& -2\alpha P_{ij}^t(t) + \ep_{ij}^2, \ \ t\ge 0,\label{cov}\\
P_{ij}^{t_k^+}(t_k) 
&=& \frac{\sigma^2_{ij}(t_k)}{P_{ij}^{t_k^-}(t_k) + \sigma^2_{ij}(t_k)} P_{ij}^{t_k^-}(t_k).\label{cov_measur}
\end{eqnarray}
The optimal filter is uniformly asymptotically stable, cf. (\ref{state_est}), and has uniformly bounded variance, since
the unconditional state variance is uniformly bounded, cf. Lemma \ref{prop_lemma}.1.

\begin{rem}[Estimation of relative skew]\label{skew_estimator}
The estimate $\hat{X}_{ij}(t)$ and its error variance $P_{ij}^t(t)$ can be used to obtain the optimal estimate for the
relative skew $a_{ij}(t)$. This is because the conditional distribution of $X_{ij}(t)$ given the measurements at time
$t$, $\mathcal{Y}_t:=\{y_{ij}(t_k)\}_{t_k\le t} \cup \{y_{ji}(t_k)\}_{t_k\le t}$, is $\mathcal{N}(\hat{X}_{ij}(t),P_{ij}^t(t))$. Hence
\begin{equation}\label{skew_estimator_formula}
\E[a_{ij}(t)|\mathcal{Y}_t] = c_{ij}(t) \E[e^{X_{ij}(t)}|\mathcal{Y}_t] = c_{ij}(t) e^{\hat{X}_{ij}(t) +
\frac{1}{2}P_{ij}^t(t)}.
\end{equation}
This estimate has uniformly bounded conditional and unconditional variance, as well.
\end{rem}
\begin{rem}[Antisymmetry property]\label{antisymmetry}
From the filtering equations one can easily check that if we define the measurement $y_{ji}(t_k) := -y_{ij}(t_k)$, and
assume that $\sigma^2_{ij}(t_k) = \sigma^2_{ji}(t_k)$, then for all $t\ge0$ we have that $\hat{X}_{ji}(t) = - \hat{X}_{ij}(t), \
P_{ij}(t) = P_{ij}(t)$. This means that the optimal estimates can be, in principle, obtained by both communicating nodes. We further study the implementation of the filter in the next section.
%
\end{rem}
\begin{thm}[Filter error variance] \label{filter_variance_thm}
The error variance of the optimal filter satisfies
\begin{eqnarray}
\bar{P} \le \frac{1}{2}[ -(1-A)(\Sigma^2 - E) + \sqrt{(1-A)^2(\Sigma^2 - E)^2 + 4(1-A)E\Sigma^2}],
\end{eqnarray}
where $E:= \frac{\ep_{ij}^2}{2\alpha}$, $A = e^{-2\alpha \bar{T}}$ and $\bar{P}, \bar{T}, \Sigma^2$ denote either the
supremum or limit superior of $P_{ij}(t), t_{k+1} - t_k, \sigma^2_{ij}(t_k)$, respectively. 
In particular,
\begin{equation}
\lim_{\sup(t_{k+1} - t_k) \to 0} \sup_{t\ge 0 } P_{ij}(t) = 0.
\end{equation}
\end{thm}
\begin{proof}
It follows from (\ref{cov}) that for $t \in [t_k,t_{k+1})$
\begin{equation}\label{var_bd1}
P_{ij}^t(t) = (P_{ij}^{t_k^+}(t_k) - E)e^{-2\alpha (t - t_k)} + E
            \le (P_{ij}^{t_k^+}(t_k) - E) A + E,
\end{equation}
since $P_{ij}(t) \le E$ for all $t\ge0$.
From (\ref{cov_measur}), we get that
\begin{equation}\label{var_bd2}
P_{ij}^{t_k^+}(t_k) = \frac{1}{\frac{1}{P_{ij}^{t_k^-}(t_k)} + \frac{1}{\sigma^2_{ij}(t_k)}}.
\end{equation}
Substituting this into (\ref{var_bd1}) and taking $\sup$ or $\overline{\lim}$ in both sides of the resulting equation
yields the inequality  $\bar{P} \le (\frac{1}{\frac{1}{\bar{P}} + \frac{1}{\Sigma^2}} - E) A + E$.
Solving this equation leads to a quadratic equation. The upper bound is obtained the largest
solution of the quadratic equation (which is positive and the quadratic is increasing at that
point).
\end{proof}

\subsection{Implementation issues}\label{implementation} 

In an actual implementation the optimal estimator is run by receiving node $j$. We point out an issue in that
node $j$ does not know the reference time evolution $t$, which is in fact the unknown it effectively seeks to estimate.
Therefore it cannot run the differential equations (\ref{state_est}), (\ref{cov}). 

To tackle this, we propose a practical fix which leads to a stable filter with uniformly bounded error variance.
Define $\tilde{\hat{X}}_{ij}(\tau) := \hat{X}_{ij}(\tau_j^{-1}(\tau)),$ where again  $t = \tau_j^{-1}(\tau)$. 
It follows directly from (\ref{state_eq}) that
\begin{equation}\label{est_tran}
\frac{d\tilde{\hat{X}}_{ij}(\tau)}{d\tau} = -\alpha\frac{1}{a_j(t)}\tilde{\hat{X}}_{ij}(\tau).
\end{equation}
Note that $a_j(t)$ is a random process, which is not measurable by $j$, therefore (\ref{est_tran}) is not an ordinary differential equation.
%
%
Let us define the suboptimal filter
\begin{equation}\label{est_tran2}
\frac{d\bar{X}_{ij}(\tau_j)}{d\tau} = -\alpha\frac{1}{f_j(\tau)}\bar{X}_{ij}(\tau),
\end{equation}
where $f_j(t)$ is measurable 
with respect to $\mathcal{Y}_{ij}^t:=\{y_{ij}(t_k)\}_{t_k\le t} \cup \{y_{ji}(t_k)\}_{t_k\le t}$ 
and  $f_j(t) > 0$ holds a.s.  For simplicity, we may set $f_j(t) = 1 = \E[a_j(t)]$ or $f_j(t) = \E[a_j(t) | \mathcal{Y}_{ij}^t]$
but the analysis holds for any a.s. positive $\mathcal{Y}_{ij}^t-$ measurable function. 
For any such $f_j(t)$ we obtain a filter
with bounded unconditional error covariance provided that, as before, the noise variance is uniformly bounded. 
%

\begin{thm}[Properties of the suboptimal filter]\label{subopt_thm}
Let the suboptimal estimator $\bar{X}_{ij}(\tau_j)$ run at node $j$
\begin{equation}\label{filter_suboptimal}
\frac{d\bar{X}_{ij}(t)}{dt} = -\frac{a_j(t)}{f_j(t)}\alpha\bar{X}_{ij}(t),
\end{equation}
where $f_j(t)$, is a $\mathcal{Y}_{ij}^t - $ measurable 
random variable such that $f_j(t) > 0, $ a.s.
At a measurement let
\begin{eqnarray}\label{filter_meaus}
\bar{X}_{ij}(t_k^+) &=& \bar{X}_{ij}(t_k^-) + k(t_k)(y_{ij}(t_k) - \bar{X}_{ij}(t_k^-)),
\end{eqnarray}
where $k(t_k)\in [0,1]$ is a $\mathcal{Y}_{ij}^{t_{k}}$ - measurable random variable, and $y_{ij}(t_k)$ is the measurement as
in Theorem~\ref{meas_thm}.
Assume that the measurement error variance sequence satisfies
\begin{equation}
\sup_k \sigma_{ij}^2(t_k) = \Sigma^2 < +\infty,
\end{equation}
Then, the filter is uniformly asymptotically stable and has bounded unconditional error variance for any selection of $f_j(t)>0$ and $\{k(t_k)\} \subset [0,1]$.
\hide{
\begin{enumerate}
\item The filter is uniformly asymptotically stable.
\item The filter has bounded unconditional error variance for any selection of $k(t_k)$,
$k(t_k) \in (0,1)$ and $\inf_k k(t_k) > 0$.
\item Selecting the gain as $k(t_k) = \min(\frac{\tilde{P}_{ij}^-(t_k)}{\tilde{P}_{ij}^-(t_k) +
\sigma_{ij}^2(t_k)}, \underline{k})$, for some $\underline{k}>0$, where
      \begin{eqnarray}
      \frac{d\tilde{P}_{ij}(\tau_j)}{d\tau_j} &=& \frac{1}{f_j(t)} (-2\alpha \tilde{P}_{ij}(\tau_j) + \ep_{ij}^2) ; \ \
\tilde{P}_{ij}^0(0) = 0,\\
      \Leftrightarrow \frac{d\tilde{P}_{ij}(t)}{dt} &=& \frac{a_j(t)}{f_j(t)}(-2\alpha \tilde{P}_{ij}(t) +
\ep_{ij}^2) ; \ \ \tilde{P}_{ij}^0(0) = 0,\\
      \tilde{P}_{ij}^{t_k^+}(t_k) &=& \frac{\sigma_{ij}^2(t_k)}{\tilde{P}_{ij}^{t_k^-}(t_k) + \sigma^2_{ij}(t_k)}
\tilde{P}_{ij}^{t_k^-}(t_k),\label{cov_meas}
      \end{eqnarray}
then for any selection of $f_j(t) \le a_j(t)$\footnote{For instance, one possible selection is to consider $f_j(t) = e^{- \frac{1}{4}\frac{\ep_j^2}{\alpha_j}}e^{\hat{X}_j(t) - cP_j(t)}$ with $c\ge3$, which gives $f_j(t) \le a_j(t)$ with probability greater or equal than $0.9985$ for given $t$.}, the filter has uniformly bounded unconditional error variance.
\end{enumerate}
}
\end{thm}

\begin{proof}
The fact that $\bar{X}_{ij}(t)$ is uniformly asymptotically stable follows directly by using the Lyapunov function $V(x) =
\frac{1}{2}x^2$, in (\ref{filter_suboptimal}). 
To establish boundedness of the unconditional error variance, note the useful inequality $\bar{P}_{ij}(t) \le 2\E[\bar{X}_{ij}^2(t)] + 2\E[X_{ij}^2(t)]$; in addition, $X_{ij}(t)$ is $L^2-$bounded and $\E[X_{ij}^2(t)]\le \frac{\ep_i^2 + \ep_j^2}{2\alpha}$ (cf. Prop. \ref{translation_prop}, and (\ref{cov_X})). 

In an interval between two measurements, $t\in [t_{k-1},t_k)$, it holds:
\begin{equation}
\bar{X}_{ij}^2(t) = \bar{X}_{ij}^2(t_{k-1})e^{-2\alpha\int_{t_k}^{t}\frac{a_j(t)}{f_j(t)}} \le \bar{X}_{ij}^2(t_{k-1});
\end{equation}
%
%
at a measurement, $\bar{X}_{ij}^2(t_k)$ is convex in the gain $k(t_k)$, whence:
\begin{align}
\E[\bar{X}_{ij}^2(t_k^+)] \le& \max(\E[\bar{X}_{ij}^2(t_k^+)], \E[X_{ij}^2(t_k)] + \sigma^2_{ij}(t_k))\nonumber\\
\le& \max(\E[\bar{X}_{ij}^2(t_k^+)], C), 
\end{align}
for $C:=  \frac{\ep_i^2 + \ep_j^2}{2\alpha} + \Sigma^2$. Both together establish that $\bar{X}_{ij}^2(t)$ is $L_2-$bounded.
\end{proof}
\begin{rem}[Gain selection] One option is to pick the gain as $k(t_k) = \min(\frac{\tilde{P}_{ij}^-(t_k)}{\tilde{P}_{ij}^-(t_k) +
\sigma_{ij}^2(t_k)}, \underline{k})$, for some $\underline{k}>0$, where
      \begin{eqnarray}
      \frac{d\tilde{P}_{ij}(\tau_j)}{d\tau_j} &=& \frac{1}{f_j(t)} (-2\alpha \tilde{P}_{ij}(\tau_j) + \ep_{ij}^2) ; \ \
\tilde{P}_{ij}^0(0) = 0,\\
      \Leftrightarrow \frac{d\tilde{P}_{ij}(t)}{dt} &=& \frac{a_j(t)}{f_j(t)}(-2\alpha \tilde{P}_{ij}(t) +
\ep_{ij}^2) ; \ \ \tilde{P}_{ij}^0(0) = 0,\\
      \tilde{P}_{ij}^{t_k^+}(t_k) &=& \frac{\sigma_{ij}^2(t_k)}{\tilde{P}_{ij}^{t_k^-}(t_k) + \sigma^2_{ij}(t_k)}
\tilde{P}_{ij}^{t_k^-}(t_k),\label{cov_meas}
      \end{eqnarray}
\end{rem}      
The following remark summarizes a discrete version of the optimal filter.
unconditional error variance.

\begin{rem}[Discrete filter]\label{discrete_filter}

If the filter does not make estimate updates between measurements, then there is no need to run a differential equation,
and therefore no implementation issue; it is a discrete Kalman filter. Based on (\ref{sol}), its  state update equation
is
\begin{eqnarray}\label{state_discrete}
X_{ij}(t_{k+1}) = e^{-\alpha(t_{k+1} - t_k)} X_{ij}(t_k) + \Gamma(k)\bar{W}_{ij}(k+1),
\end{eqnarray}
where $\bar{W}_{ij}(k)$, is a Gaussian white noise sequence and  $\Gamma(k) := \frac{\ep_{ij}}{\sqrt{2\alpha}}\sqrt{1 -
e^{-2\alpha(t_{k+1} - t_k)}}$. Also,
\begin{eqnarray}
\hat{X}_{ij}^{t_{k+1}^-}(t_{k+1}) &=& e^{-\alpha(t_{k+1} - t_k)} \hat{X}_{ij}^{t_k^+}(t_k),\\
P_{ij}^{t_k^-}(t_k) &=& e^{-2\alpha(t_{k+1} - t_k)} P_{ij}^{t_{k-1}^+}(t_{k-1}) + \Gamma^2(k),\\
P_{ij}^{t_k^+}(t_k) &=& (1-k(t_k))^2P_{ij}^{t_k^-}(t_k) + k(t_k)^2\sigma^2_{ij}(t_k),
\end{eqnarray}
whence, using an inductive argument,
\begin{equation}
\sup_{t\ge 0}P_{ij}^{t}(t) \le \frac{\ep_{ij}^2}{2\alpha}.
\end{equation}
For a practical implementation, we could approximate $t_{k+1} - t_k$ by $\tau_j(t_{k+1}) - \tau_j(t_k)$ if $i,j \ne
0$.
\end{rem}

\section{Network-wide smoothing of pairwise estimates}\label{smoothing}
In this section we make a connection with distributed asynchronous smoothing of pairwise estimates.For illustration, we adopt the asynchronous Jabobi approach of~\cite{solis,barooah,smoothing}, but other approaches, such as the Randomized Kaczmarz~\cite{RKO,Rand_Kaczmarz}, which has showcased better accuracy and energy savings. 

We begin the section by presenting a general lemma, and then comment on how this can be used to develop a network-wide
offline state estimator. We denote the $L^2$-norm by $||\cdot||_2$.

\begin{lem}[A generalized least-squares problem]\label{rls}
Given square integrable random vectors $X, Y \in L^2(P)$, where $X\in\mathbb{R}^m$, consider the problem of estimating
$X$ by $A^Tv$ where $A\in\mathbb{R}^{m\times n}$ is a non-random matrix with $(AA^T)$ invertible, and $v\in
\mathbb{R}^n$ is $Y-$measurable and square-integrable. If the error criterion is $\E[\|X - A^Tv\|_2^2 | Y]$, then the
unique solution is
\begin{equation}
v  = (AA^T)^{-1}A\E[X | Y],
\end{equation}
which is the same as solving the deterministic least-squares problem for $v$ with error criterion  $\|\E[X|Y] -
A^Tv\|_2^2$.
\end{lem}

\begin{proof}
By the orthogonality principle, we have $\E[\|X - A^Tv\|_2^2 | Y] = \E[\|X - \E[X|Y]\|_2^2 | Y] + \E[\|\E[X|Y] -
A^Tv\|_2^2 | Y]$, where the first term is independent of $v$, and the second one is $\|\E[X|Y] - A^Tv\|_2^2$.
Minimizing the second term over $v$ yields the result.
\end{proof}

Let us denote the directed graph of links across which state estimates are available at some given time\footnote{We drop
the time indices for ease of notation.} by $G = (V,E)$. We assume that the graph is connected. Let us also denote the
reduced incidence matrix of the graph \cite{graph} obtained from the incidence matrix by removing the row corresponding
to node $0$, by $A$. 
Then $AA^T$ is the principal submatrix of the \emph{Laplacian} of the graph \cite{graph}, which is known to
be positive
definite \cite{arvind} for connected graphs, and hence invertible.  Since $X_{ij} = X_j - X_i$ and $X_0 = 0$,
\begin{equation}\label{constraints}
X = A^Tv,
\end{equation}
where $X = \{X_{ij}\}\in \mathbb{R}^{|E|}$, $A\in\mathbb{R}^{|E|\times n}$, $v = \{X_i\}_1^n \in \mathbb{R}^n$.
Consider now the network-wide estimation problem as one of determining the MMSE estimate $v$, with error criterion
\begin{equation}
\E[\|X - A^Tv\|_2^2 | Y],
\end{equation}
where $Y :=\{y_{ij}\}$ denotes the $\sigma-$algebra  generated by the measurements obtained till the current time
instant, abusing notation and dropping the time indices. The unique solution is $v  = (AA^T)^{-1}A\E[X | Y],$ which
implies that we need to have access to the MMSE estimates of $X_{ij}$ given \emph{all} link measurements. However,
because the Brownian motions $W_{ij},W_{kl}$ are dependent if $i$ or $j\in \{k,l\}$, these MMSE estimators are obtained
based on measurements on \emph{at least} \footnote{In fact the optimal filter derived in section \ref{net_synch_sec}
shows that they might depend on measurements on \emph{all} links.} all links adjacent to $i,j$. Therefore, they are not
the same as the pairwise link estimates $\hat{X}_{ij}$ developed in section \ref{pairwise_sync}, and  cannot be used to
obtain the MMSE estimate $v$. A simple fix is to redefine the error criterion, with component-wise conditioning, as
\begin{equation}
\sum_{(i,j)\in E} \E[(X_{ij} - (A^Tv)_{ij})^2 | Y_{ij}],
\end{equation}
whence the MMSE estimate becomes $\bar{v} = (AA^T)^{-1}A\hat{X}$, where $\hat{X} = \{\hat{X}_{ij}\}$.
To solve this, we propose using the efficient asynchronous distributed algorithm of \cite{solis,arvind}.

Setting the $i-th$ derivative of $F(v)$ to 0, yields $(AA^T)_iv - A_i\hat{X} = 0$, and a straightforward analysis
\cite{arvind} shows that coordinate descent gives rise to an asynchronous scheme where node $i$ updates its estimate
$v_i$ according to
\begin{equation}\label{spatial_smoothing}
v_i = \frac{1}{d_i} \sum_{j: (i,j)\in E \mbox{ or } (j,i)\in E} (v_j + \hat{X}_{ji}),
\end{equation}
where $d_i$, denotes the total degree of node $i$. Note that the scheme is fully distributed and uses no information
about the network topology.

For our problem, spatial smoothing can be used to obtain
\begin{enumerate}
  \item Nodal offset estimates from relative offset estimates since $\tau_{ij} = \tau_j - \tau_i$.
  \item Nodal skew estimates from relative skew estimates since $\log{a_{ij}} = \log{a_j} - \log{a_i}$.
  \item Nodal state estimates from relative state estimates since $X_{ij} = X_j - X_i$.
\end{enumerate}

We note in passing that the synchronous version of the exact same scheme \cite{arvind} can be obtained by using
stochastic approximation \cite{borkar}. In order to solve the system $A^Tv - \hat{X} = 0$, or equivalently $AA^Tv -
A\hat{X} = 0$ we define the stochastic approximation scheme
\begin{equation}
v_{k+1} = v_k - a_k(AA^Tv_k - A\hat{X}),
\end{equation}
since the limiting ODE \cite{borkar} is $\dot{v}(t) = -AA^Tv + A\hat{X}$. In \cite{arvind}, they substitute $a_k$ by a
diagonal matrix $D = diag(\frac{1}{d_1},\hdots,\frac{1}{d_n})$ with entries the inverses of the relative degrees.

\section{Network clock synchronization} \label{net_synch_sec}

We now develop an asynchronous \footnote{Note that since the goal is to synchronize clocks, a synchronous algorithm can
only be implemented in a centralized fashion}
algorithm for the optimal network state estimation problem given by the continuous-discrete Kalman-Bucy filter
\cite{jaswinski} for the state vector $X_t = \{X_i(t)\}\in \mathbb{R}^n$. We use the convention that all vectors are
column vectors. Let us also define node $i$'s neighborhood $\mathcal{N}_i := \{j\in V : j=i, \hbox{ or
}(i,j)\in{E},\hbox{ or } (j,i)\in{E}\}$.  The following theorem summarizes the optimal continuous-discrete Kalman-Bucy
filter equations \cite{jaswinski}.

\begin{thm}[Network optimal state estimation]\label{centr_filter}
~\\
Suppose that nodes $j=0,\hdots n$ make noisy measurements $y_{ij}$ of $X_{ij}$ over links $(i,j)\in E$, in an
asynchronous fashion such that
\begin{equation}\label{measurements}
y_{ij}(t_k) = X_{ij}(t_k) + v_{ij}(t_k) = M(t_k)^TX(t_k) + v_{ij}(t_k),
\end{equation}
where $M(t_k)\in \mathbb{R}^n$, $v_{ij}(t_k) \sim \mathcal{N}(0,\sigma^2_{ij}(t_k))$, and
\begin{eqnarray}\label{M}
( M(t_k) )_m &=& \left\{
                 \begin{array}{ll}
                   -1, & m=i, \\
                   1, & m=j, \\
                   0, & \hbox{else.}
                 \end{array}
               \right.
\end{eqnarray}
Then $P_t^t = P_t$, i.e., the conditional and the unconditional covariance coincide. Between measurements, the optimal
filter is given by
\begin{eqnarray}
\frac{d\hat{X}_t}{dt} &=& -\alpha \hat{X}_t, \ \ \ \hat{X}_0 = 0,\label{state_est2}\\
\frac{dP_t}{dt} &=& -2\alpha P_t + E^2, \ \ \  P_0 = 0_{n\times n},\label{cov_X2}\\
E &=& diag(\ep_1^2,\hdots \ep_n^2).
\end{eqnarray}

At a measurement, the estimate is updated as follows:

\begin{eqnarray}
(\hat{X}_{t_k^+})_m &=& (\hat{X}_{t_k^-})_m + K(t_k)  (y_{ij}(t_k) - \hat{X}_j(t_k^-) +
\hat{X}_i(t_k^-)),\label{meas_state}\\
K(t_k) &=& \frac{P_{mj}^{t_k^-} - P_{mi}^{t_k^-}}{c_k} \label{net_kalman_gain}\\
(P^{t_k^+})_{ml} &=& (P^{t_k^-})_{ml} - \frac{(P_{mj}^{t_k^-} - P_{mi}^{t_k^-})(P_{jl}^{t_k^-} -
P_{il}^{t_k^-})}{c_k},\label{meas_cov}
\end{eqnarray}
where $c_k := P_{ii}^{t_k^-} + P_{ij}^{t_k^-} - 2P_{ij}^{t_k^-}+ \sigma_{ij}^2(t_k)$. It is uniformly asymptotically
stable with uniformly bounded covariance.
\end{thm}
It is straightforward to extend this scheme to account for the case that two or more measurements are taken at the same
time instant $t_k$.

Nodal skew estimates $\hat{a}_{i}$ and relative skew estimates $\hat{a}_{ij}$ can be obtained by means of
\begin{eqnarray}
\hat{a}_i(t) &=& c_i(t) e^{\hat{X}_i(t) + \frac{1}{2}P_{ii}^t(t)}, \label{nodal_skew_est_net}\\
\hat{a}_{ij}(t) &=& c_{ij}(t) e^{\hat{X}_j(t) - \hat{X}_i(t) + \frac{1}{2}(P_{ii}^t(t) + P_{jj}^t(t)
-2P_{ij}^t(t))}\label{rel_skew_est_net}.
\end{eqnarray}

\subsection{Protocol considerations} \label{protocols}

The optimal state estimator \emph{cannot} be implemented as is in a decentralized fashion. Except for the implementation
issues discussed in Section \ref{implementation}, (\ref{state_est2}) and (\ref{cov_X2}) can be implemented in a
decentralized fashion such that, for instance, node $i$ updates the state estimate $(\hat{X}_t)$ and the $i-$th row of
the covariance matrix $P_t$. However, the same is not true at a measurement. In fact, for a connected graph (the
situation considered here) \emph{all} entries may be non-zero, which, in turn, implies,  cf. (\ref{meas_state}),
(\ref{meas_cov}), that the filter updates at a measurement cannot be run in a decentralized fashion. The rationale for
this is as follows: Just before the first measurement, the covariance matrix is diagonal, and after the first
measurement, say at link $(i,j)$, all entries $p_{kl}$, where $k \hbox{ or } l \in {i,j}$ are updated. Even if we assume
that the covariance matrix is local (i.e., $p_{ml} = 0$, if $(m,l)\not \in E$) then this won't be necessarily the case
after a measurement at link $(i,j)$ since, as is evident from (\ref{meas_cov}), the entry $p_{ml}$ will be updated even
if $(m,l)\not \in E$, provided that $m,l\in \mathcal{N}_i\cup\mathcal{N}_j$. In particular, nodes that are two hops away
will update their corresponding covariance entry. In fact, for a connected graph, there exists a series of measurements
after which all entries of the covariance matrix and the state estimator are updated at a measurement in a link.

Therefore, the optimal state estimation algorithm is, by nature, centralized and there would need to be some message
passing in order for all nodes to agree on the error covariance matrix which is used in determining the Kalman gains.

Another issue is that relative skew estimates do not have the symmetry property, i.e., $\hat{a}_{ij} \ne
\frac{1}{\hat{a}_{ji}}$ in general as is evident from (\ref{skew_estimator_formula}); in fact, from
(\ref{skew_estimator_formula}) and Remark \ref{antisymmetry}, it follows that
\begin{equation}
\hat{a}_{ij}(t)\hat{a}_{ji}(t) = e^{P_{ij}^t(t)} \ge 1.
\end{equation}

\subsection{A distributed filter for skew estimation}\label{distr_imp}

We now describe a distributed suboptimal filter for network-wide estimation of skews. In order to design a distributed
algorithm, we impose a constraint on the structure of the gain $K(t_k)$. In particular, on a measurement at link
$(i,j)$, we allow $(K(t_k))_m \ne 0$ only if $m\in\{i,j\}$, which implies the constraint that only the state estimates
$\hat{X}_i,\hat{X}_j$ will be updated. Then,
\begin{equation}
P^{t_k^+} = (I - K(t_k)M(t_k)^T)P^{t_k^-}(I - K(t_k)M(t_k)^T)^T + \sigma_{ij}^2(t_k)K(t_k)K(t_k)^T.
\end{equation}
where $M(t_k)$ is as in (\ref{M}). The elements of the covariance matrix are updated as follows: If $m,l \not \in
\{i,j\}$ then $p^+_{ml} = p^-_{ml}$. If $m \not \in \{i,j\}$, then
\begin{eqnarray} \label{upd1}
P^+_{mi} &=& (1+k_i)P^-_{mi} - k_iP^-_{mj},\\
P^+_{mj} &=& k_jP^-_{mi} +(1- k_j)P^-_{mj}.
\end{eqnarray}
Finally
\begin{eqnarray} \label{upd2}
P^+_{ii} = (1+k_i)((1+k_i)P^-_{ii} - k_iP^-_{ij}) - k_i( (1+k_i)P^-_{ij} - k_iP^-_{jj}) + k_i^2\sigma_{ij}^2,\\
P^+_{ij} = (1+k_i)(k_jP^-_{ii} +(1- k_j) P^-_{ij}) - k_i( k_jP^-_{ij} +(1-k_j)P^-_{jj}) + k_ik_j\sigma_{ij}^2,\\
P^+_{jj} = k_j(k_jP^-_{ii} + (1-k_j)P^-_{ij}) + (1-k_j) (k_jP^-_{ij} + (1-k_j)P^-_{jj}) + k_j^2\sigma_{ij}^2,
\end{eqnarray}
where $k_i,k_j$ are the $i-$th and $j-th$ entries of $K(t_k)$. Note that the only diagonal entries of $P$ that are
updated are $P^+_{ii},P^+_{jj}$. So minimizing the trace over $k_i, k_j$ boils down to minimizing $P^+_{ii}$ over $k_i$,
and $P^+_{jj}$ over $k_j$. This yields $k_i = -\frac{b_i}{2a}$, $k_j = -\frac{b_j}{2a}$ for
\begin{eqnarray}
a &:=& P^-_{ii} + P^-_{jj} - 2P^-_{ij} + \sigma_{ij}^2 > 0,\\
b_i &:=&2(P^-_{ii} - P^-_{ij}),\\
b_j &:=&2(P^-_{ij} - P^-_{jj}),
\end{eqnarray}
from which it follows that $p_{ii},p_{jj}$ decrease at a measurement, i.e.,
\begin{eqnarray}
P^+_{ii} = -\frac{b_i^2}{4a} + P^-_{ii}, \label{pii}\\
P^+_{jj} = -\frac{b_j^2}{4a} + P^-_{jj}.
\end{eqnarray}

\begin{rem}
From (\ref{upd1}), it follows that $P_{ml}$ might be updated even if $(m,l) \not \in E$ , as was the case in
the
centralized filter. However, this does not constitute a problem for implementation, because the gains depend only on
entries of the covariance matrix which are either diagonal or correspond to links.
\end{rem}

\begin{rem}
It follows from (\ref{pii}), that $P_{ii}^+ = \frac{[P_{ii}^-(P_{jj}^- + \sigma_{ij}^2)] - (P_{ij}^-)^2}{(P_{ii}^- +
P_{jj}^- + \sigma_{ij}^2) - 2P_{ij}^-}$. If we assume that $P_{ii}^+ \le \frac{P_{ii}^-(P_{jj}^- +
\sigma_{ij}^2)}{P_{ii}^- + P_{jj}^- + \sigma_{ij}^2 }$, e.g. if $P_{ij}^- \le 0$, or $P_{ij}^- \approx 0$ then we are in
the setup of Theorem \ref{filter_variance_thm} and we can apply the analysis therein to derive variance bounds by
substituting $\Sigma^2$ by $P_{j} + \Sigma^2$ where $P_{j}, \Sigma^2$ denote upper bounds on $P_{jj},
\sigma_{ij}^2(t_k)$.
\end{rem}

\section{Offset estimation} \label{offset_estimation_sec}

We have studied the problem of optimal estimation of the state (logarithms of skews) and the skews. However, clock
synchronization ultimately amounts to estimating nodal offsets with respect to a reference clock \cite{nfrer_TAC}.
Developing a continuous-discrete filter for the estimation of $\tau_i(t) - t$ is problematic because of the dependency
of the time displays $\tau_i(t)$ on the skew $a_i(t)$ in (\ref{disp_eq}). Under the assumption of unknown delays, it was
shown in \cite{scott} that the relative offset between two clocks cannot be estimated unless delays are assumed to be
symmetric, or have a known affine characterization of asymmetry \cite{nfrer_TAC}. For a link $(i,j)$, we will assume that
delays, as measured in reference clock units, are symmetric. 
We will use two packets, one sent from node $i$ to node $j$, and the other sent from node $j$ to node $i$, together with
an estimate of the pairwise skew (as obtained by the pairwise filters in section \ref{pairwise_sync}) in order to make
an estimate of the relative offset $\tau_{ij} := \tau_j - \tau_i$. The estimate can be obtained in a similar way as done
in \cite{nfrer_TAC}, given an estimate of the relative skews $\hat{a}_{ij}(k),\hat{a}_{ji}(k)$ (cf. (\ref{relative_skew})),
by ignoring skew variations between the sent times of the two packets. From Figure \ref{ping2} and using the notation
therein we get:
\begin{eqnarray}
r_{ij}(k) &=& s_i(k) + \tau_{ij}(t_1) + d_{ij},\\
r_{ji}(k) &=& s_j(k) + \tau_{ji}(t_3) + d_{ji},\\
\tau_{ij}(t_3) &=& \tau_{ij}(t_1) + (s_j(k) - r_{ij}(k) + d_{ij})(1 - \hat{a}_{ji}(k)),\\
\tau_{ij}(t_4) &=& \tau_{ij}(t_1) + (r_{ji}(k) - s_i(k))(\hat{a}_{ij}(k) - 1),
\end{eqnarray}
where $d_{ij},d_{ji}$ are the delays as measured by the local clocks, $j,i$ respectively. Using $\tau_{ji}(t_3) =
-\tau_{ij}(t_3)$ and $d_{ij} = \hat{a}_{ij}d_{ji}$ we get:
\begin{eqnarray}
\hat{\tau}_{ij}(k) &=& -r^{(k)}_{ji} + (s^{(k)}_j + \hat{a}_{ij}\hat{d}_{ji}),\label{offset_est}\\
\hat{d}_{ij}(k) &:=&  \frac{1}{2\hat{a}_{ji}}[(r^{(k)}_{ji} - s^{(k)}_j) + (r^{(k)}_{ij}-s^{(k)}_i) + (s^{(k)}_j -
r^{(k)}_{ij})(1-\hat{a}_{ji})].\label{delay_est}
\end{eqnarray}

\begin{figure}
   \centering
 \includegraphics[totalheight=0.2\textheight]{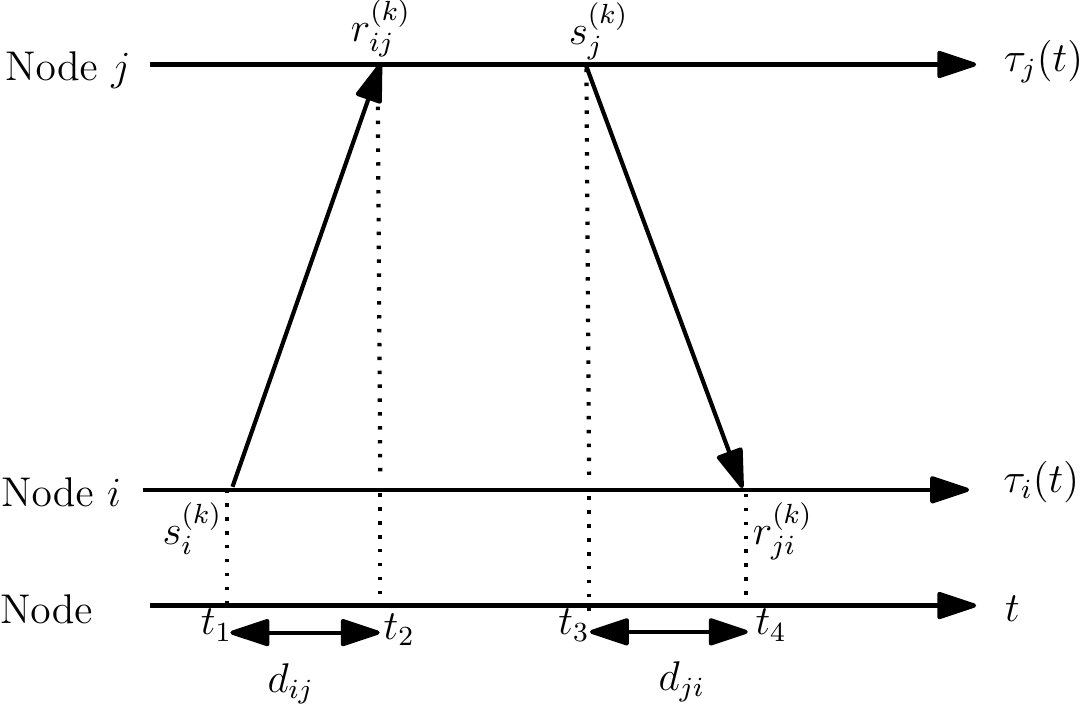}
  \caption{\small Message exchanges between two nodes for offset estimation}\label{ping2}
\end{figure}

An analysis of (\ref{offset_est}) based on stochastic Taylor approximation
carried out in a similar way as was done in section \ref{meas_sec}, so as to bound the approximation error.

Using this scheme, pairwise estimates of the relative offsets can be obtained. Then, the spatial smoothing algorithm
\cite{solis,arvind} presented in section \ref{smoothing} can be used for the offline smoothing of the pairwise
estimates, since relative offsets also satisfy constraints of the form $\tau_{ij} = \tau_j - \tau_i$.

\subsection{Performance evaluation of clock synchronization algorithms}\label{predict}

Because of link delays, it is impossible for a node to have instantaneous access to another node's clock display, and we
need the scheme of the previous section to estimate relative offsets. The relative offset estimates cannot be used as a
practical metric for clock synchronization performance. A metric of performance can be derived based on the fact that
(cf. Figure \ref{ping})
\begin{equation}
\bar{a}_{ij}(t_k) = |\frac{r_{i,j}^{(k+1)} - r_{i,j}^{(k)}}{s_i^{(k+1)} - s_i^{(k)}}|,
\end{equation}
if the link delays of the two packets (as measured in clock $i$'s units) are assumed to be the same constant. We use the
$\bar{a}_{ij}(t_k)$ to denote the average relative skew in the interval $[s_i^{(k)},s_i^{(k+1)}]$ of clock $i$. Node $i$
can predict the receipt time of the second packet, $s_i^{(k+1)}$, using $s_i^{(k)}, r_{i,j}^{(k)}, s_i^{(k+1)}$ and an
estimate of $\bar{a}_{ij}(t_k)$. This can be then compared to the actual receipt time to obtain a metric of clock
synchronization.

\section{Model-based Clock synchronization protocol (MBCSP)} \label{MBCSP}

In this section, we combine our previous analysis to present the specifications of a proposed model-based distributed
clock synchronization protocol. We consider the same abstraction regarding the network topology as in Section
\ref{net_synch_sec}.

Nodes can exchange time-stamped packets with their neighbors in an asynchronous fashion. We consider two modes of
communication:
\begin{enumerate}
  \item \textbf{Skew estimation:} \ 
A node $i$ sends two packets with minimal time separation to a neighboring node $j$ as
explained in Section \ref{meas_sec} and depicted in Figure \ref{ping}. Node $i$ includes its parameter $\ep_i$ in the
second packet sent to node $j$, along with its state estimate $\hat{X}_i$ and $p_{ii},p_{ij}$.
  \item \textbf{Offset estimation:} \ 
A node $i$ sends a packets to node $j$. Upon receipt, node $j$ sends a packet to node $i$.
Node $i$ includes its parameter $\ep_i$,  along with its state estimate $\hat{X}_i$ and $p_{ii},p_{ij}$ in the first
packet sent to node $j$. This is depicted in Figure \ref{ping2}.
\end{enumerate}
In both modes, the sending node time-stamps and includes in the packet the time (according to its local clock) that it
sends a packet while the receiving node time-stamps the time (according to its local clock) that it receives a packet.

An arbitrary node in the network, say node $i$, is required to store parameters $\alpha, \ep_i$, as well as
\begin{enumerate}[(i)]
  \item Its state estimate $\hat{X}_i$.
  \item The $i$-th row of the covariance matrix, $\{P_{ij}\}_{j=1}^{n}$.
  \item Its nodal offset estimate $\widehat{\tau_i - t}$.
  \item For any neighboring node, say node $j$, the relative offset estimate $\hat{\tau_{ij}}$.
  \item The last time, according to its local clock, that it performed a state estimate update (along with an update of
the $i-$th row of the covariance matrix), based on either (\ref{state_est2}), (\ref{cov_X2}) or (\ref{meas_state}) and
the covariance update equations of Section \ref{distr_imp}. We denote this time by $u_i$.
\end{enumerate}

Each node, is responsible for calling four routines:
\begin{enumerate}
  \item \textbf{Skew update:} \ Upon receipt of the second packet in the \emph{skew estimation} mode, say from node $i$ to node $j$
(see Figure \ref{ping}), node $j$ has all four time-stamps $s_i^{(k)}, s_i^{(k+1)}, r_{i,j}^{(k)}, r_{i,j}^{(k+1)}$, as
well as $\alpha,\ep_i,\ep_j$. Node $j$ makes a measurement according to (\ref{link_meas}); if node $i$ is node 0 then
$t_k = s_i^{(k)}$, otherwise $r_{i,j}^{(k+1)}$ can be used instead. Before node $i$ sends the second packet, it updates
its state estimate along with the $i-$th row of the covariance matrix according to
      \begin{eqnarray}
      \hat{X}_i^{t_{k+1}^-}(t_{k+1}) &=& e^{-\alpha\Delta t_k} \hat{X}_i^{t_k^+}(t_k),\label{st1}\\
      P_{ij}^{t_k^-}(t_k) &=& e^{-2\alpha\Delta t_k} P_{ij}^{t_{k-1}^+}(t_{k-1}) + \frac{\ep^2_j}{2\alpha}(1 -
e^{-2\alpha\Delta t_k}) \mathbb{I}_{j=i}\label{c1},
      \end{eqnarray}
      where $\mathbb{I}$ denotes the characteristic, function and $\Delta t_k:=t_{k+1} - t_k$; it can be approximated by
the difference $s_i^{(k+1)} - u_i$. Node $i$ includes its updated state $\hat{X}_i$ as well as the updated values for
$P_{ii}, P_{ij}$. Finally, node $i$ sets $u_i = s_i^{(k+1)}$. Upon receipt of the second packet, node $j$ updates its
state estimate along with the $j-$th row of the covariance matrix according to (\ref{st1}), (\ref{c1}) where now $\Delta
t_k$ can be approximated with $r_{i,j}^{(k+1)} - u_j$, and then $u_j$ is updated to $r_{i,j}^{(k+1)}$.
      Then, node $j$ calculates the state and covariance updates after the measurement based on (\ref{meas_state}),
along with the covariance update equations of Section \ref{distr_imp}. Node $j$ can send an ACK packet to node $i$ and
$i$ performs the exact same tasks as $j$.

  \item \textbf{Relative offset update:} \  Before the sender, say node $i$, sends the first packet in the \emph{offset estimation}
mode to a receiver, say node $j$, (see Figure \ref{ping2}), node $i$ updates its state estimate $\hat{X}_i$ and the
$i-$th row of the covariance matrix according to (\ref{st1}), (\ref{c1}) where now $\Delta t_k$ can be approximated with
$s_i^{(k)} - u_i$, and then $u_i$ is updated to $s_i^{(k)}$. Node $i$ includes in this packet $\hat{X}_i, P_{ii},
P_{ij}$. Upon receipt of this packet, node $j$ also updates its state estimate $\hat{X}_j$ and the $j-$th row of the
covariance matrix according to (\ref{st1}), (\ref{c1}) where now $\Delta t_k$ can be approximated with $r_{i,j}^{(k)} -
u_j$, and then $u_j$ is updated to $r_{i,j}^{(k)}$. Then, node $j$ can estimate $\hat{a}_{ij},\hat{a}_{ji}$ based on
(\ref{rel_skew_est_net}), where $t$ is approximated by $s_i^{(k)}$ and $r_{i,j}^{(k)}$, respectively. A symmetric
estimate can then be calculated by $\hat{a}_{ij} \leftarrow \sqrt{\frac{\hat{a}_{ij}}{\hat{a}_{ji}}}$, and
$\hat{a}_{ji} = \frac{1}{\hat{a}_{ij}}$.

      Upon receipt of the second packet in the \emph{offset estimation} mode, from node $j$ to node $i$ (see Figure
\ref{ping2}), node $i$ has all four time-stamps $s_i^{(k)}, r_{i,j}^{(k)}, s_j^{(k)}, r_{j,i}^{(k)}$ along with the
estimate $\hat{a}_{ij}$, so it can perform offset and delay estimation according to (\ref{offset_est}) and
(\ref{delay_est}). The delay estimate needs to be non-negative so we set $\hat{d}_{ij} = max(\hat{d}_{ij},0)$, and
$\hat{d}_{ij} = \hat{a}_{ij}\hat{d}_{ji}$. Node $i$ can send an ACK packet to node $j$ in order to agree on an estimate
$\hat{\tau}_{ji}(k) = - \hat{\tau}_{ij}(k)$.

  \item \textbf{Spatial smoothing:} \ Upon the completion of the \emph{relative offset update}, nodes $i$ and $j$ perform a spatial
smoothing update based on (\ref{spatial_smoothing}) where $v_i$ represents $\widehat{\tau_i - t}$ and $\hat{X}_{ji}$
represents $\hat{\tau}_{ji}$.

  \item \textbf{Time prediction:} \ Upon completion of the \emph{skew update} task (after the receipt of the second packet in the
\emph{skew estimation} mode, say from node $i$ to node $j$ (see Figure \ref{ping}) ) node $i$ can compute the predicted
receipt time of the second packet ,$r_{i,j}^{(k+1)}$, if the ACK packet contains $r_{i,j}^{(k+1)}$ based on the
discussion in Section \ref{predict}. This requires a calculation of $\hat{a}_{ij}$, which can be performed in the exact
same way as was done in \emph{relative offset update}. Node $i$ can then obtain the difference of the predicted receipt
time to the actual value $r_{i,j}^{(k+1)}$. The same difference is computable by node $j$.
\end{enumerate}

We refer to this protocol as Model-based Clock synchronization protocol (MBCSP). We have implemented MBCSP in
Matlab with the graph abstraction for the network topology (cf. Section \ref{meas_sec}). In order to make the scenario
realistic with the interference in wireless networks, we have also implemented a simple Medium Access Control (MAC)
\cite{robert} mechanism, where two interfering transmissions collide and packets get dropped. Without serious loss of
generality, we use the primary interference model \cite{robert}, in which a node cannot be transmitting a packet to  more
than one node and cannot be receiving a packet by more than one node. In addition, a node is not allowed to be
transmitting if it is receiving a packet. These constraints require the activation set, defined as the set of links
activated at each time, to be a matching \cite{graph}.

We test MBCSP against the protocol of \cite{solis}, which will be hereafter referred to as Spatial Smoothing (SS). In
this protocol, relative skews are calculated based on only link relative skew measurements using exponential forgetting.
Nodal skew estimates are obtained using spatial smoothing as explained in Section \ref{smoothing}. Relative offset
estimates are obtained using the formulas in Section \ref{offset_estimation_sec}, and are consequently smoothed using
spatial smoothing to obtain nodal offset estimates.

We have also implemented a protocol that performs relative skew estimation using the pairwise filter of
Section \ref{pairwise_sync} instead of the distributed network filter of Section \ref{distr_imp}. Nodal skew estimation
is accomplished using the spatial-smoothing algorithm (cf. Section \ref{smoothing}) while offset estimation is
carried out in the exact same way as MBCSP. We refer to this as \emph{Hybrid} protocol. Its performance is
intuitively expected to lie in between that of SS and MBCSP, which we validate in the next section.

\section{Simulations}\label{sec:simulations}

we present simulation results of the model properties and the performance of the clock synchronization
protocol.

In Figure \ref{sim1} we present a simulation of the skew and time display of a clock with $\alpha_i = 10, \ep_i = 1$ for
30 reference time units, using Matlab. The fluctuations in the instantaneous skew give rise to a time-varying offset.

\begin{figure}
  \centering
  \includegraphics[totalheight=0.3\textheight]{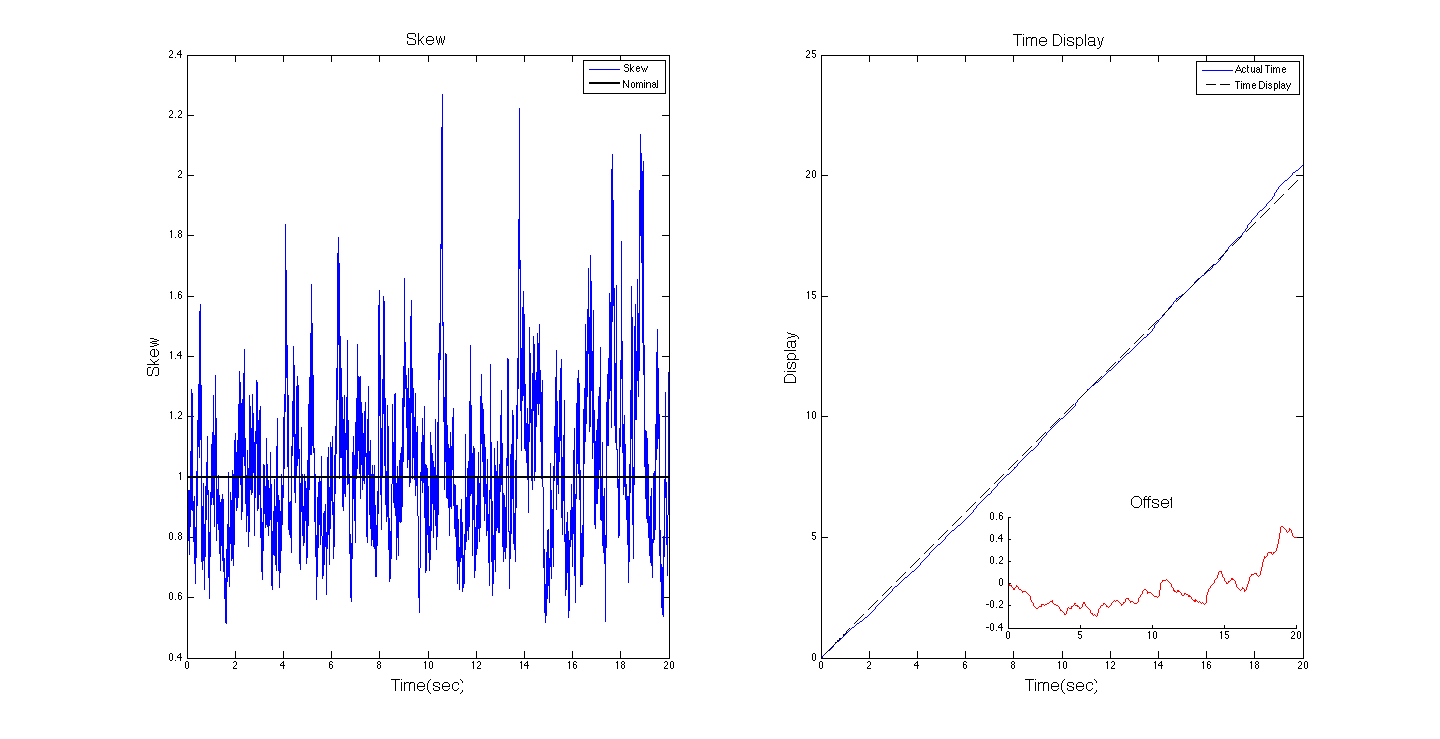}
  \caption{\small Instantaneous skew and time display for a clock with $\frac{\alpha_i}{\ep_i^2} = 10$.}\label{sim1}
\end{figure}

In Figure \ref{sim1} we present a simulation of the clock display variance (cf. (\ref{var_bound})) and its lower bound
(cf. (\ref{var_lb})) for two different clocks, with parameters $\alpha_i = 10, \ep_i = 1$ and $\alpha_i = 10, \ep_i =
10$ for 20 reference time units. The upper bound (cf. (\ref{vartime_bound})) is not displayed because it appears to
vastly overestimate the variance. Both the variance and its lower bound appear to grow \emph{linearly} with time; in
fact we applied a curve-fitting approach with a nominal curve $ax^b$, and the exponent was estimated close to $1$ in
both cases with small mean-square error (MSE).

\begin{figure}[ht]\label{variance_fig}
\centering
\subfigure{
\includegraphics[scale=0.5]{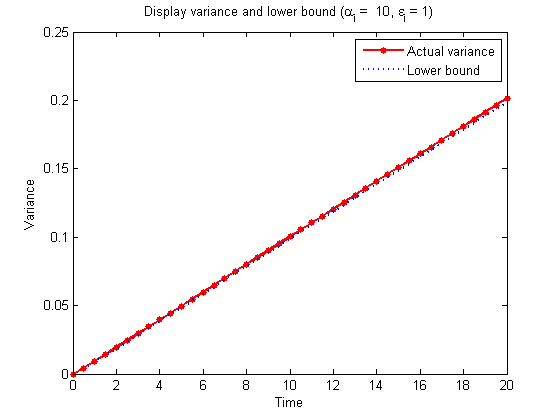}
}
\subfigure{
\includegraphics[scale=0.5]{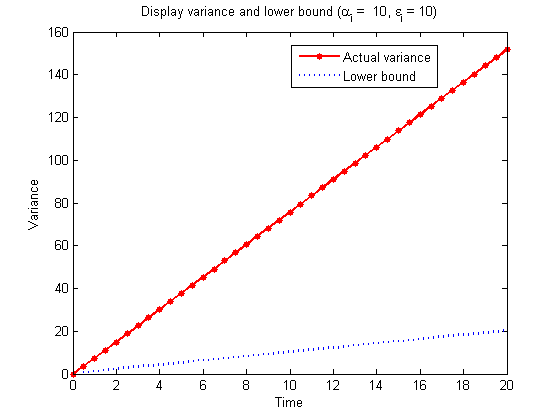}
}
\caption{\small Display variance and lower bound for two different clocks.}
\end{figure}

In Figure \ref{allan_sim} we illustrate the Allan variance of a clock with the same parameters as above and a comparison
with the performance index of Remark \ref{asd}. Unlike the index of Remark \ref{asd}, the Allan variance is not an
increasing function. However, the approximation is good for small values of $T$.

\begin{figure}
  \centering
  \includegraphics[totalheight=0.3\textheight]{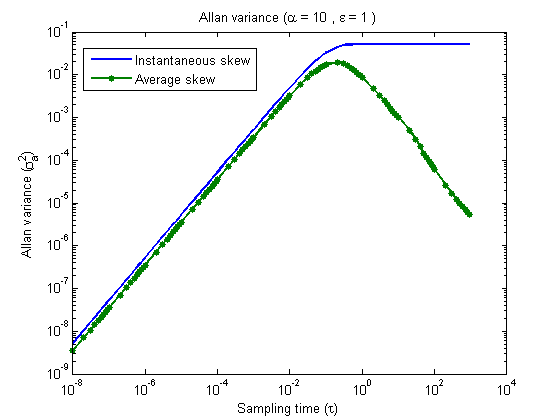}
  \caption{\small Allan variance vs asymptotic skew difference variance for a clock with $\frac{\alpha_i}{\ep_i^2} =
10$.}\label{allan_sim}
\end{figure}

In Figure \ref{allan_fit}, we present the measured Allan variance for a Berkeley mote clock \cite{solis}, as well as the
best fit for the Allan variance of our model (\ref{allan_model} using Matlab; the parameters corresponding to the best
fit were $\hat{\alpha}_i = 66.4, \hat{\ep}_i = 4.15\cdot10^{-5}$. It is evident that the model can only capture a decay
of the Allan variance with $\tau$ but is unable to capture the fact that Allan variance appears to increase for $\tau
\ge 60s$.
However, note that the scale is logarithmic and the average absolute error of the fit is $4.4659\cdot10^{-10}$.

\begin{figure}
  \centering
  \includegraphics[totalheight=0.3\textheight]{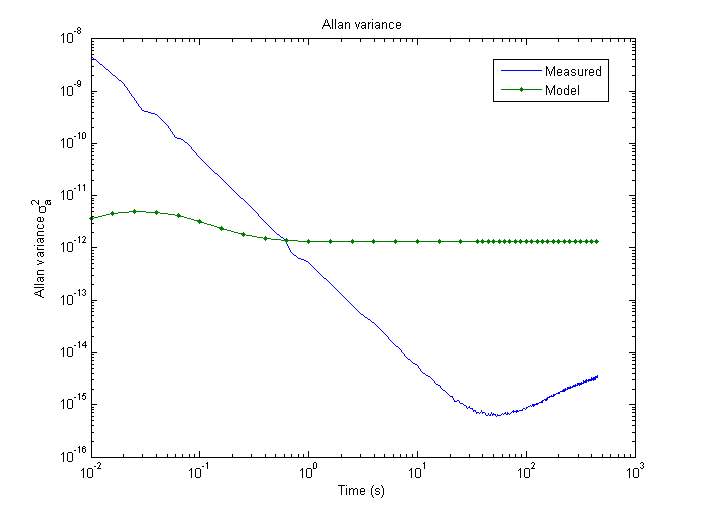}
  \caption{\small Parameter estimation using Allan variance.}\label{allan_fit}
\end{figure}

We have studied the performance degradation of the distributed filter of Section \ref{distr_imp} as opposed to the
optimal centralized Kalman filter (cf. Section \ref{net_synch_sec}) for various network topologies. In all cases, the
variance of the suboptimal filter was very close to the optimal variance. In Figure \ref{filt_var_fig}, we present the
average state variance (trace of the covariance measurement divided by the size of the state) for a linear network with
10 clocks with parameters $\alpha = 10, \ep=1$. Measurements are taken every $T = 0.002$ times units, and for each measurement,
one link is selected at random. The performance degradation is defined as the ratio of the distributed filter variance
to the variance of the optimal Kalman filter.

\begin{figure}
  \centering
  \includegraphics[totalheight=0.3\textheight]{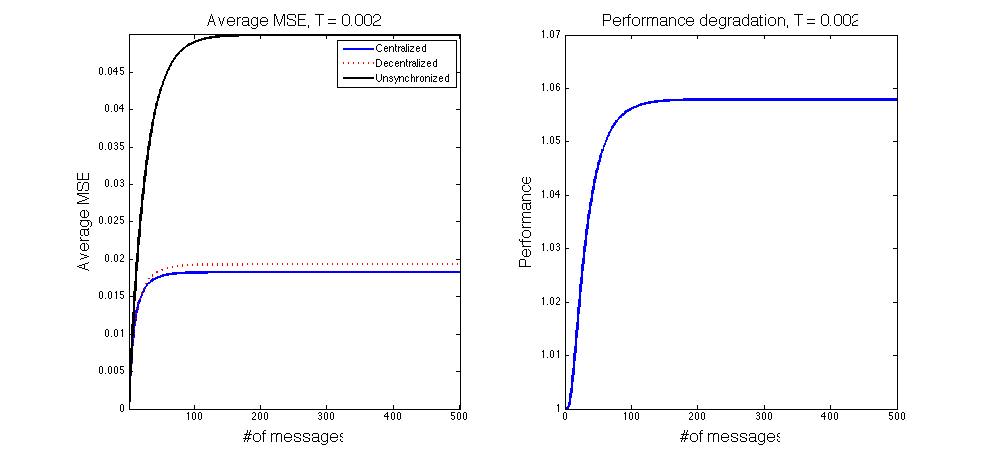}
  \caption{\small Performance degradation of distributed state estimation filter for a linear network with 10 nodes and
random periodic measurements.}\label{filt_var_fig}
\end{figure}

We have simulated the performance of MBCSP 
against both SS \cite{solis} and the Hybrid scheme defined in the previous section. The
results in Figure \ref{synch_sim} were obtained for 2 clocks and 5 clocks respectively. The clock parameters are set to
$\alpha=10,\ep=1$. The time precision, i.e., the sampling time of the numerical simulation (cf. remak \ref{OU_sim}), is
set to $\Delta t = 10^{-5}$. Delays are generated as random variables independent, uniformly distributed with mean
$500\Delta t$. For skew estimation, the two packets of Figure \ref{ping} are sent with a separation of $40 \Delta t$,
while the second packet for offset estimation (cf. Figure \ref{ping2}) is sent $20 \Delta t$ after the first one is
received. Measurements are performed at random times and random links with frequency $1|E|$ for skew estimation, and
$6|E|$ for offset estimation, where the unit time is defined as $\frac{1}{\Delta t}$ mini-slots, and $|E|$ is the number
of directed links in the network. We have conducted the simulations for $120$ time units and different realizations of
the white noises. The results of Figure \ref{synch_sim} illustrate mean absolute errors (MAE); we use three indices,
namely the nodal skew estimation MAE, nodal offset estimation MAE, and the MAE in predicting receipt times (cf. Section
\ref{predict}). The columns labeled as ``No sync'' are used to present the mean absolute real values of the
corresponding quantities.

\begin{figure*}[ht]


\subfigure{
\begin{scriptsize}
\begin{tabular}{|c| c c c c| c c c c| c c c|}
        \hline
        {\bfseries Node} & {\bfseries Offset} & & & & {\bfseries Skew} & & & & {\bfseries Prediction} & &  \\
	& No sync & SS & Hybrid  & MBCSP & No sync & SS & Hybrid & MBCSP & SS & Hybrid & MBCSP\\
        \hline
1 & 0.00000 & 0.00000 & 0.00000 & 0.00000 & 1.00000 & 0.00000 & 0.00000 & 0.00000 & 0.01657 & 0.01359 & 0.01364\\
2 & 0.57687 & 0.00131 & 0.00120 & 0.00119 & 1.00548 & 0.20866 & 0.16970 & 0.17014 & 0.01657 & 0.01359 & 0.01364\\
\hline
\end{tabular}
\end{scriptsize}
}

\subfigure{
\begin{scriptsize}
\begin{tabular}{|c| c c c c| c c c c| c c c|}
        \hline
        {\bfseries Node} & {\bfseries Offset} & & & & {\bfseries Skew} & & & & {\bfseries Prediction} & &  \\
	& No sync & SS & Hybrid  & MBCSP & No sync & SS & Hybrid & MBCSP & SS & Hybrid & MBCSP\\
        \hline
1 & 0.00000 & 0.00000 & 0.00000	& 0.00000 & 1.00000 & 0.00000 &	0.00000 & 0.00000 & 0.04115 & 0.02609 & 0.02655\\
2 & 0.60233 & 0.02310 & 0.02308 & 0.02307 & 1.02628 & 0.21036 & 0.18741 & 0.17734 & 0.04113 & 0.02610 & 0.02656\\
3 & 0.29349 & 0.02044 & 0.02038 & 0.02038 & 1.01466 & 0.22604 & 0.17117 & 0.16513 & 0.04111 & 0.02609 & 0.02655\\
4 & 1.71270 & 0.02049 & 0.02043 & 0.02043 & 0.98548 & 0.21835 & 0.16800 & 0.16808 & 0.04111 & 0.02610 & 0.02654\\
5 & 0.48743 & 0.02193 & 0.02193 & 0.02191 & 1.00192 & 0.21063 & 0.16870 & 0.16165 & 0.04115 & 0.02609 & 0.02654\\
\hline
\end{tabular}
\end{scriptsize}
}

\caption{\small Performance evaluation of clock synchronization algorithms based on data obtained from MATLAB
simulations.}\label{synch_sim}
\end{figure*}

We have performed numerous simulations for various parameters and network topologies. In all cases, we observed that
Hybrid and MBCSP track the clock skews significantly better than the ad-hoc exponential forgetting scheme of SS. This
results in more accurate relative offset estimates, delay estimates predicted receipt times. However, since both schemes
use the Spatial Smoothing algorithm (cf. Section \ref{smoothing}) to obtain nodal offset estimates from relative offset
estimates, their accuracy in nodal offset estimates is comparable. We also observed that the Hybrid
scheme performs significantly better than SS but worse than MBCSP\footnote{Note that in the case of only two nodes
(pairwise synchronization) the offest estimation accuracy is, trivially, the same.}.

We have obtained real clock data from two Berkeley motes \cite{solis} exchanging time-stamps in the two communication
modes described in Section \ref{MBCSP}, namely skew and offset estimation. Based on the time-stamp collection we
performed a trace-driven simulation \footnote{Note that trace-driven simulation is equivalent to an actual
implementation, since our protocols use only the acquired time-stamps and have minimal computational complexity.} to
obtain a comparative evaluation of the three scehmes. Clock parameters $\alpha_2,\ep_2$ were estimated
based on Allan variance, as shown above, to be $\hat{\alpha}_i = 66.4, \hat{\ep}_i = 4.15\cdot10^{-5}$. The accuracy
was set to $ 1 \mu s$. The results are presented in Figure \ref{synch_em}. In this case, we cannot define the real
skews and offsets so we present the offset and skew estimates obtained from the three different schemes, in mean
absolute value. However, we can still define the prediction MAE which is the metric of performance for
clock synchronization algorithms (cf. Section \ref{predict}).  The
prediction MAE is very low in all cases, in the order of tens of $\mu s$, which means that the receipt times are
precisely predicted within the clock accuracy of $1 \mu s$ in many cases. Our schemes yield a $45\%$ decrease in the
prediction MAE.

\begin{figure*}[ht]
\begin{scriptsize}
 \begin{tabular}{|c| c c c| c c c| c c c|}
        \hline
        {\bfseries Node} & {\bfseries Offset} & & & {\bfseries Skew} & & & {\bfseries Prediction} & &  \\
	& SS & Hybrid  & MBCSP & SS & Hybrid & MBCSP & SS & Hybrid & MBCSP\\
        \hline
1 & 0.00000  & 0.00000 & 0.00000 & 1.00000 & 1.00000 & 1.00000 & 6.74515e-07 & 3.84218e-07 &
3.84218e-07\\	
2 & 453.82502 & 453.82502 & 453.82502 & 0.99996 & 1.00000 & 1.00000 & 6.74599e-07 & 3.84207e-07 & 3.842070e-07	\\
\hline
\end{tabular}
\end{scriptsize}

\caption{Trace-driven simulation of clock synchronization algorithms based on data obtained from Berkeley motes.}\label{synch_em}
\end{figure*} 
\section{Conclusion}

We have developed a mathematical model for the skews and the time displays of different clocks and analyzed its
properties. The instantaneous skews given by the model have expected value 1 at all times, while their variance is
bounded. Additionally, time displays are unbiased, but, nonetheless, their variance grows with time, which makes the
synchronization problem challenging. We have calculated the Allan variance \cite{allan} of the model. It was shown that
if a different clock is taken as reference, the time displays of all other clocks can be expressed with respect to it
using the same model, with different parameters and a change of time scales. We have developed and analyzed a method to
obtain noisy state measurements on a link, and used these measurements to develop a continuous-discrete Kalman-Bucy
filter for the pairwise estimation of the logarithm of skews. The differential equations of the pairwise estimator are
not readily implementable since they require integration with respect to the unknown reference time, so we have proposed
an implementable stable filter which has uniform bounded unconditional variance. The analysis of the pairwise estimation
was applied to handle the network-wide state estimation in three ways. First, an off-line algorithm for the filtering of
pairwise estimates was proposed, and it was shown that using the distributed scheme of \cite{solis}, \cite{arvind} for
the smoothing of pairwise estimates is optimal for a particular selection of error criterion. In addition, the optimal
linear filtering equations were derived for the network case, which give rise to an online asynchronous centralized
scheme which is stable and of bounded variance. We have also described an efficient distributed suboptimal scheme.
Finally, we have presented a scheme to estimate relative offset estimates based on estimates of the relative skews, and
suggested the spatial smoothing algorithm of \cite{solis} to obtain nodal offset estimates from relative offset
estimates. We have implemented our protocol in Matlab and have conducted a simulation study that shows increase of
performance compared to \cite{solis} for the model under study. We have also performed trace-driven simulation based on
time-stamps obtained by Berkeley motes. Our scheme outperforms the one in \cite{solis} by $45\%$, where we used the
accuracy in predicting receipt time-stamps as synchronization metric.

%
%
\bibliographystyle{IEEEtran}
\bibliography{references}

\appendices
\renewcommand{\theequation}{A-\arabic{equation}}
\section{Proof of Lemma \ref{prop_lemma}}\label{prop_app}
By applying It\^{o}'s formula~\cite{jaswinski} to the function $f(X_i(t),t) := X_i(t)e^{\alpha_i t}$
we have
\begin{equation}
df(X_i(t),t) = \ep_ie^{\alpha_it}dW_i(t).
\end{equation}
Hence, the solution of (\ref{OU}) satisfies
\begin{equation}\label{sol}
X_i(t) = \ep_i\int_0^t \! \! e^{-\alpha_i(t - s)}dW_i(s),
\end{equation}
which, in turn, shows that $X_i(t)$ is a zero-mean Gaussian random variable. The variance can be computed by use of
It\^{o}'s isometry~\cite{jaswinski}:
\begin{equation}\label{cov_X}
\E[X_i(t)^2] = \ep_i^2\int_0^te^{-2\alpha_i(t - s)}ds =
\frac{\ep_i^2}{2\alpha_i}(1 - e^{-2\alpha_it}).
\end{equation}
This  further implies that
\begin{eqnarray}
\E[a_i(t)] &=& \E[e^{X_i(t)}]e^{-\frac{1}{4}\frac{\ep_i^2}{\alpha_i}(1 - e^{-2\alpha_it})} = 1,\\
Var(a_i(t)) &=& c_i^2(t)\E[e^{2X_i(t)}]-1\nonumber\\
            &=& e^{\frac{\ep_i^2}{2\alpha_i}(1-e^{-2\alpha_it})} - 1 \nearrow \ e^{\frac{\ep_i^2}{2\alpha_i}} - 1.\label{varskew_bound}
\end{eqnarray}
By Tonelli's theorem~\cite{royden}
\begin{equation}
\E[\tau_i(t)] = \int_0^t\!\!\E[a_i(t')]dt' = t.
\end{equation}
It follows from (\ref{sol}) and Tonelli's theorem that
\begin{equation}
Var(\tau_i(t)) = \int_0^t \! \! \int_0^t\!\!\E[a_i(r)a_i(s)]ds\,dr - t^2.
\end{equation}
Note that $\E[a_i(r)a_i(s)] = c_i(r,s)\E[e^{X_i(r)+X_i(s)}]$, where $c_i(r,s) :=e^{-\frac{\ep_i^2}{4\alpha_i}(2-e^{-2\alpha_ir}-e^{-2\alpha_is})}$ and $X_i(r) + X_i(s) \thicksim \mathcal{N}(0,\sigma_i^2(r,s))$ with $\sigma_i^2(r,s) := -2\log c_i(r,s) + 2\E[X_i(r)X_i(s)]$. From (\ref{sol}) and It\^{o}'s isometry we have 
\begin{eqnarray}\label{cor_X}
\E[X_i(r)X_i(s)] &=& \ep_i^2e^{-\alpha_i(r+s)}\E[\int_0^r \! \! e^{\alpha_it} dW_i(t)\ \ \int_0^s \! \! e^{\alpha_it}
dW_i(t)]\\
 &=& \ep_i^2e^{-\alpha_i(r+s)}\int_0^{s\wedge r} \! \! e^{2\alpha_it}dt =
\frac{\ep_i^2}{2\alpha_i}e^{-\alpha_i(r+s)}(e^{2\alpha_i s\wedge r} - 1),
\end{eqnarray}
where $a\wedge b := min(a,b)$. Therefore
\begin{equation}\label{skew_cor}
\E[a_i(r)a_i(s)] = e^{\frac{\ep_i^2}{2\alpha_i}e^{-\alpha_i(r+s)}(e^{2\alpha_i s\wedge r} - 1)},
\end{equation}
which, in turn, yields
\begin{eqnarray}
& Var(\tau_i(t)) =  2\int_0^t \! \! \int_0^r\!\!e^{f(s,r)}ds\,dr - t^2, \label{var_bound}\\
\mbox{where } & f(s,r):= 
\frac{\ep_i^2}{2\alpha_i}e^{-\alpha_ir}(e^{\alpha_is} - e^{-\alpha_is})\label{fsr}.
\end{eqnarray}
Evaluating the integral (\ref{var_bound}) analytically is not possible, so we study the asymptotic behavior of the
variance. For $r\le s$, $f(s,r)$ is strictly decreasing in $r$, and strictly increasing in $s$,  therefore $f(s,r) \le
f(s,s)$. Also note that $g(s):=f(s,s)$ is strictly increasing.
This implies that
\begin{equation}\label{vartime_bound}
Var(\tau_i(t)) < (e^\frac{\ep_i^2}{2\alpha_i}-1)t^2 = O(t^2).
\end{equation}
To obtain a lower bound on $Var(\tau_i(t))$ we use Jensen's inequality~\cite{royden} to get 
\begin{eqnarray}\label{var_lb}
Var(\tau_i(t)) &\ge& (A - 1)t^2,\\
A &:=& 
e^{\frac{2}{t^2}\int_0^t \! \! \int_0^r\!\!{f(s,r)}ds\,dr} = e^{h(t)},\nonumber\\
h(t) &:=& \frac{1}{t^2} \frac{\ep_i^2}{\alpha_i^2}  (t + \frac{1}{2\alpha_i}(1-e^{-2\alpha_it}) - \frac{2}{a_i}(1-e^{-\alpha_it})) = \Omega(1/t).\nonumber
\end{eqnarray}
The rest follows from $(e^{\Omega(\frac{1}{t})} - 1)t^2 = \Omega(t)$.
\hfill$\blacksquare$
\renewcommand{\theequation}{B-\arabic{equation}}
\section{Proof of Theorem~\ref{meas_thm}}\label{model_just_app}
\begin{defn}
For stochastic variables, $y(t),x(t)$ depending on a parameter $t$, whenever we write $y(t) = x(t) + O(t^k)$, for some
real number $k$, this is interpreted 
in the $L^2-$norm, 
that is to say there exists a $c>0$, such that $\E[(y(t) - x(t))^2] \le c t^{2k}$, for all $t \in
(0,1)$, in particular $\overline{\lim}_{t\to 0^+} \frac{\E[(y(t) - x(t))^2]^\frac{1}{2}}{t^{k}}
< \infty $. The extensions for $\Theta(\cdot),\Omega(\cdot),o(\cdot)$ and for the vector case are straightforward.
\end{defn}
An application of Cauchy-Schwartz inequality~\cite{royden} yields that $O(t^k)O(t^l)
= O(t^{k+l})$, while Minkowski's inequality~\cite{royden} yields $O(t^k) + O(t^l) = O(t^{k\wedge l})$.
The following result shows how to perform a first-order stochastic Taylor expansion for the processes under study, namely, skews and time displays.

\begin{lem}[Stochastic Taylor's expansion] \label{lemma_appr}
For $\delta t_k := t_{k+1} - t_k $, and  $d_k$ a positive random variable independent of $W_i$:
\begin{eqnarray}
\tau_i(t_{k+1}) &=& \tau_i(t_k) + a_i(t_k)\delta t_k + O(\delta t_k^\frac{3}{2}), \label{taylor1}\\
\tau_i(t_k + d_k) &=& \tau_i(t_k) + a_i(t_k)d_k + O(\E[d_k^3]^\frac{1}{2}),\label{taylor2}\\
a_i(t_{k+1}) &=& a_i(t_k) + O(\delta t_k^\frac{1}{2}).\label{taylor3}
\end{eqnarray}
\end{lem}
\begin{proof}
Let $a(a_i(t),t):= (-\alpha_i \log a_i(t) + \frac{3\ep_i^2}{4}(1-e^{-2\alpha_i t}))a_i(t)$. From (\ref{disp_eq}), (\ref{sde_a}) we have
%
\begin{eqnarray}\label{taylor}
\tau_i(t_{k+1}) &=& \tau_i(t_k) + \int_{t_k}^{t_{k+1}} a_i(t)dt\nonumber\\
&=& \tau_i(t_k) + a_i(t_k)\delta t_k + \int_{t_k}^{t_{k+1}} \int_{t_k}^{t'}
a(a_i(s),s) ds dt' + \int_{t_k}^{t_{k+1}}\int_{t_k}^{t'} \ep_ia_i(s)dW_i(s)dt
\end{eqnarray}
Define $A := \int_{t_k}^{t_{k+1}} \! \! \int_{t_k}^{t'} a(a_i(s),s) ds dt', B := \int_{t_k}^{t_{k+1}} \!
\! \int_{t_k}^{t'} \ep_ia_i(s)dW_i(s)dt'$. It is easy to verify that there exits $C_1>0$ s.t. $|a(a_i(t),t)| \le C_1(a_i(t) + a_i^2(t))$. By the fact that $a_i(t)$ is $L^p-$ bounded for all $p \ge 1$ (cf. Lemma \ref{bounded}) 
and Tonelli's theorem~\cite{royden}, $A = O(\delta t_k^2)$. 
By It\^{o}'s isometry~\cite{jaswinski} and Fubini's theorem~\cite{royden}, we get for some
constant $C>0$
\begin{eqnarray}
\E[B^2] &=& \ep_i^2 \int_{t_k}^{t_{k+1}}\int_{t_k}^{t_{k+1}} \E[\int_{t_k}^{t'} \int_{t_k}^{t''}
a_i(s)a_i(r)dW_i(s)dW_i(r)]dt'dt''\nonumber\\
&\le& C \int_{t_k}^{t_{k+1}}\int_{t_k}^{t_{k+1}} t'\wedge t'' dt'dt''\nonumber\\
&=& O(\delta t_k^3).
\end{eqnarray}
Using the fact that $O(\delta t_k^2) + O(\delta t_k^\frac{3}{2}) = O(\delta t_k^\frac{3}{2})$
for small values of $\delta t_k$ establishes (\ref{taylor1}).
Applying the same analysis to the conditional expectation given $d_k$ yields $\E[(\tau_i(t_k + d_k) - \tau_i(t_k) -
a_i(t_k)d_k)^2 |d_k]  = O(d_k^3)$. Taking expectations in both sides establishes (\ref{taylor2}).
The proof of
(\ref{taylor3}) follows along the same lines using (\ref{sde_a}), It\^{o}'s isometry and the $L^p-$boundedness of
$a_i(t)$.
\end{proof}

From (\ref{disp_eq}), we have the following approximation
\begin{eqnarray}
r_{i,j}^{(k)} &=& \tau_j(t_k) + a_j(t_k)d_{ij}^{(k)} + e^{(r)}_k \label{meas1},\\
r_{i,j}^{(k+1)} &=& \tau_j(t_{k+1}) \!+\! a_j(t_{k+1})d_{ij}^{(k+1)} \!\! + \!\! e^{(r)}_{k+1}, \label{meas2}\\
a_j(t_{k+1})d_{ij}^{(k+1)} &=& a_j(t_k)d_{ij}^{(k+1)} + e^{(a)}_k\label{skew_delay},\\
s_i^{(k+1)} - s_i^{(k)} &=& a_i(t_k)\delta t_k + e^{(s)}_k,\label{eqn1}\\
\tau_j(t_{k+1}) - \tau_j(t_k) &=& a_j(t_k)\delta t_k + e^{(j)}_k\label{eqn2}.
\end{eqnarray}
From the previous lemma the error terms satisfy $e^{(r)}_k = O(d_{ij}^{(k)^\frac{3}{2}}), \  
e^{(a)}_k = O(d_{ij}^{(k+1)}) O(\delta t_k^\frac{1}{2}), \ 
e^{(s)}_k = O(\delta t_k^\frac{3}{2}), \ 
e^{(j)}_k = O(\delta t_k^\frac{3}{2})$. 
%
%
We adopt the following assumptions on the communication delays:


\begin{ass}[Assumption on delays]\label{delay_ass} 
The delays  $\{d_{ij}^{(k)}\}$ are positive random variables upper bounded by $D$, independent across time instants $k$,
and links $(i,j)$, and identically distributed across time instants $k$ for a fixed link $(i,j)$, as well as independent of
 $\{W_i(t)\}$.
\end{ass}
%
Subtracting (\ref{meas1}) from (\ref{meas2}) and using (\ref{skew_delay}), (\ref{eqn2}) we
get
\begin{equation}\label{rec_dif_taylor}
r_{i,j}^{(k+1)} - r_{i,j}^{(k)} = a_j(t_k)\delta t_k + a_j(t_k)(d_{ij}^{(k+1)}-d_{ij}^{(k)}) + O(\delta t_k^\frac{3}{2})
+ O(d_{ij}^{(k)^\frac{3}{2}}) + O(d_{ij}^{(k+1)^\frac{3}{2}}) + O(d_{ij}^{(k+1)}) O(\delta t_k^\frac{1}{2}) .
\end{equation}
The goal is to express the Taylor expansion for the fraction $\frac{r_{i,j}^{(k+1)} - r_{i,j}^{(k)}}{s_i^{(k+1)} - s_i^{(k)}}$, i.e., 
divide (\ref{rec_dif_taylor}) by (\ref{eqn1}). 

\begin{equation}
\frac{1}{s_i^{(k+1)} - s_i^{(k)}} = \frac{1}{(t_{k+1}-t_k)}\frac{1}{\frac{1}{(t_{k+1}-t_k)}\int_{t_k}^{t_{k+1}}a_i(t)
dt} \le
\frac{1}{(t_{k+1}-t_k)^2} \int_{t_k}^{t_{k+1}} \frac{1}{a_i(t)}dt,
\end{equation}
where we have used Jensen's inequality. In particular, it follows that for every $m$, $\delta s_i^{(k)^m} := s_i^{(k+1)} - s_i^{(k)} = \Theta(\delta t_k^m)$; 
the sender can directly control $\delta s_i^{(k)}$. 
By the Taylor expansion of $g(x) = \frac{1}{c+x},\ c>0$ around $x=0$ and (\ref{eqn1}):
\begin{equation}
\frac{1}{s_i^{(k+1)} - s_i^{(k)}} = \frac{1}{a_i(t_k)\delta t_k} + \frac{O(\delta t_k^{\frac{3}{2}})}{\Theta(\delta
s_i^{(k)^2})} = \frac{1}{\delta t_k}\frac{1}{a_i(t_k)} + O(\delta t_k^{-\frac{1}{2}}),
\end{equation}
whence
\begin{eqnarray}
\frac{r_{i,j}^{(k+1)} - r_{i,j}^{(k)}}{s_i^{(k+1)} - s_i^{(k)}} &=& \frac{a_j(t_k)}{a_i(t_k)}(1 + \frac{d_{ij}^{(k+1)} -
d_{ij}^{(k)}}{\delta t_k}) + e_k,\\
e_k &=& O(\delta t_k^\frac{1}{2}) + O(d_{ij}^{(k+1)}) O(\delta t_k^{-\frac{1}{2}}) +
[O(d_{ij}^{(k)^\frac{3}{2}}) + O(d_{ij}^{(k+1)^\frac{3}{2}})]O(\delta t_k^{-1}),
\end{eqnarray}
or, equivalently, (by the first-order Taylor expansion of $\log (c+x), \ c>0$)
\begin{eqnarray} \label{error_terms}
\log|\frac{r_{i,j}^{(k+1)} - r_{i,j}^{(k)}}{s_i^{(k+1)} - s_i^{(k)}}| &=& - \frac{1}{4}(\frac{\ep_j^2}{\alpha} -
\frac{\ep_i^2}{\alpha})(1-e^{-2\alpha t_k}) + X_{ij}(t_k) + \log|1 + \frac{d_{ij}^{(k+1)} - d_{ij}^{(k)}}{t_{k+1} -
t_k}| + e_k',\\
e_k' &=& O(\delta t_k^\frac{1}{2}) + O(d_{ij}^{(k+1)}) O(\delta t_k^{-\frac{1}{2}}) +
[O(d_{ij}^{(k)^\frac{3}{2}}) + O(d_{ij}^{(k+1)^\frac{3}{2}})]O(\delta t_k^{-1}),
\end{eqnarray}

The quantity $y_{ij}(t_k) := \log|\frac{r_{i,j}^{(k+1)} - r_{i,j}^{(k)}}{s_i^{(k+1)} - s_i^{(k)}}|$ can be obtained from measurements. Note that because of
delay variation, it might be that $r_{i,j}^{(k+1)} < r_{i,j}^{(k)}$, even though $s_i^{(k+1)} > s_i^{(k)}$, so the use
of absolute value in the definition of $y_{ij}(t_k)$ is indispensable; this corresponds to out-of order delivery of packets in wireless. 
The quantity $v_{ij}(t_k) := \log|1 +
\frac{d_{ij}^{(k+1)} - d_{ij}^{(k)}}{t_{k+1} - t_k}| +  e_k'$ 
is random, and depends on the rate of delay variations and the skew variations during the send times of the two
consecutive packets. Hence we get
\begin{equation}
y_{ij}(t_k) := X_{ij}(t_k) + v_{ij}(t_k).
\end{equation}

It is well known~\cite{jaswinski}, that linear filtering is the only computationally tractable estimation scheme for the
state equation (\ref{state_eq}). Hence, we need to model $v_{ij}(t_k)$ as white Gaussian noise, i.e., $v_{ij}(t_k) \sim
\mathcal{N}(0,\sigma^2_{ij}(t_k))$. 
Note that the random variables $\{e_k'\}$ are not independent across time, because $a_{ij}(t), X_{ij}(t)$ are not
independent increment processes. However, if we further make the assumption that link delays can be made
small (for a practical scheme to reduce delays via proper time-stamping see~\cite{FTSP}), in
particular assuming that $d_{ij}^{(k)^m} = O (\delta t_k^m), m=1,\frac{3}{2}$
implies that $e_k' = O(\delta t_k^\frac{1}{2})$. Therefore we ignore the correlation between $e_k'$ for different values of $k$, as it can 
be controlled by $\delta t_k$. 
%
By the assumption of small delays, it is reasonable to assume that the variance sequence is uniformly upper bounded and the variance 
$\sigma^2_{ij}(t_k)$ is an increasing function of $\delta t_k$, of the known quantity $\delta s_i^{(k)}$. 
%
\hfill$\blacksquare$

\end{document}